\documentclass[11pt,fleqn,a4paper]{article} 

\usepackage{amsmath,amssymb,amsthm,enumerate} 
\usepackage[mathscr]{eucal}
\usepackage{graphicx}

\usepackage{hyperref}
\hypersetup{colorlinks, linkcolor=blue, citecolor=blue, urlcolor=blue}

\newcommand{\p}{\partial}

\newcommand{\const}{{\rm const}}

\newcommand{\rank}{\mathop{\rm rank}\nolimits}

\newcommand{\spanindex}{{\mbox{\tiny$\langle\,\rangle$}}}

\newtheorem{theorem}{Theorem}

\newtheorem{corollary}[theorem]{Corollary}
\newtheorem{conjecture}[theorem]{Conjecture}
\newtheorem{proposition}[theorem]{Proposition}
\newtheorem*{proposition*}{Proposition}
{\theoremstyle{definition}

\newtheorem{remark}[theorem]{Remark}

}

\newcommand{\todo}[1][\null]{\ensuremath{\clubsuit}}

\newcommand{\noprint}[1]{}

\newcounter{mcasenum}

\begin{document}
\par\noindent {\LARGE\bf
Lie symmetries of two-dimensional shallow water\\ equations with variable bottom topography\par}

\vspace{4mm}\par\noindent{\large
Alexander Bihlo$^\dag$, Nataliia Poltavets$^\dag$ and Roman O.\ Popovych$^\ddag$
}

\vspace{4mm}\par\noindent{\it
$^\dag$\,Department of Mathematics and Statistics, Memorial University of Newfoundland,\\
$\phantom{^\dag}$\,St.\ John's (NL) A1C 5S7, Canada
}

\vspace{2mm}\par\noindent{\it
$^\ddag$\,Fakult\"at f\"ur Mathematik, Universit\"at Wien, Oskar-Morgenstern-Platz 1, A-1090 Wien, Austria%
\\
$\phantom{^\ddag}$Institute of Mathematics of NAS of Ukraine, 3 Tereshchenkivska Str., 01024 Kyiv, Ukraine
}

\vspace{4mm}\par\noindent
E-mails:
abihlo@mun.ca,
natalkapoltavets@gmail.com,
rop@imath.kiev.ua

\vspace{6mm}\par\noindent\hspace*{8mm}\parbox{140mm}{\small\looseness=-1
We carry out the group classification of the class
of two-dimensional shallow water equations with variable bottom topography
using an optimized version of the method of furcate splitting.
The equivalence group of this class is found by the algebraic method.
Using algebraic techniques,
we construct additional point equivalences between some of the listed cases of Lie-symmetry extensions,
which are inequivalent up to transformations from the equivalence group. 
}\par\vspace{2mm}

\noprint{

76-XX   Fluid mechanics {For general continuum mechanics, see 74Axx, or other parts of 74-XX}
 76Mxx  Basic methods in fluid mechanics [See also 65-XX]
  76M60   Symmetry analysis, Lie group and algebra methods
86-XX  Geophysics [See also 76U05, 76V05]
 86Axx	Geophysics [See also 76U05, 76V05]
  86A05  Hydrology, hydrography, oceanography [See also 76Bxx, 76E20, 76Q05, 76Rxx, 76U05]
35-XX   Partial differential equations
  35A30   Geometric theory, characteristics, transformations [See also 58J70, 58J72]
  35B06   Symmetries, invariants, etc.

}

\section{Introduction}

The shallow water equations are among the most studied models in geophysical fluid dynamics.
Due to the ability 
of modeling both slow moving (Rossby) and fast moving (gravity) waves,
the shallow water equations provide an ideal test bed for the development of new numerical approaches
to be used for future numerical models in geophysical fluid dynamics,
see e.g. \cite{aech15a,brec19a,brec18a,cott14a,dele10a,flye12a,salm07a}.

Besides playing an important role as an intermediate-complexity model for designing new numerical approaches for weather and climate modeling,
the shallow water equations are still routinely used in research and operational tsunami propagation models~\cite{brec18a,tito97Ay,tito95a,tito98a}.
A main challenge arising in the application of the shallow water equations in ocean wave propagation is the need to incorporate a variable bottom topography,
since the height of the water column over the bottom of the ocean basin determines the phase speed of gravity waves.

\looseness=-1
Owing to the considerable interest in the shallow water equations, 
there is a large body of literature devoted to finding exact solutions and conservation laws for these equations.
Both are important for numerical considerations
since exact solutions can be used for benchmarking numerical methods
and conservation laws can be applied to checking the reliability of new numerical schemes.
The construction of exact solutions and conservation laws is already a challenging problem 
without considering variable bottom topographies. 
With variable bottom topography, exact solutions are mostly known for simple profiles, 
such as linear slopes~\cite{carr58a} or parabolic bowls~\cite{thac81a}.

Below, we briefly review some of the existing works on the one- and the two-dimensional shallow water equations 
within the framework of group analysis of differential equations, 
which are related to the present paper and include the computation of exact solutions and conservation laws. 
The main point of division in the various studies in this regard is whether Lagrangian or Eulerian coordinates are~used.

The two-dimensional shallow water equations
(as well as the semi-geostrophic equations that arise in meteorology and oceanography)
in Lagrangian coordinates over flat bottom topography were considered
from the point of view of the group analysis of differential equations, e.g., in~\cite{bila06Ay}.
Lie symmetries and certain potential and variational symmetries were computed therein,
and variational symmetries were used for finding first-order conservation laws according to Noether's theorem.
The one-dimensional shallow water equations over flat bottom topography were considered in~\cite{siri16a}
in both Lagrangian and Eulerian coordinates as a limit case of the Green--Naghdi model,
and their first-order conservation laws in Lagrangian coordinates with no counterparts among conservation laws in Eulerian variables were found.
Special first-order conservation laws for the one-dimensional case over variable bottom topography in Lagrangian coordinates
were constructed in~\cite{akse20a} using their relation to hydrodynamic conservation laws
of a potential system for the shallow water equations in Eulerian variables.

\looseness=-1
In Eulerian variables,
Lie symmetries of the two-dimensional shallow water equations with parabolic bottom topography were computed in~\cite{levi89By}
and then used for finding exact solutions via classifying classical Lie reductions.
It was shown in~\cite{chir14a} that in one dimension,
the shallow water equations with a linear bottom topography can be mapped
to the shallow water equations with flat bottom topography by a point transformation.
Classical symmetry analysis of a modified system of one-dimensional shallow water equations,
including the construction of the maximal Lie invariance algebra of this system
and the classification of invariant solutions, was carried out in~\cite{szat14a}.
Lie symmetries and zeroth-order conservation laws
in the one-dimensional case with variable bottom topography were described in~\cite{akse16a}, 
and the geometric structure of self-similar solutions of the second kind for the flat bottom topography 
was studied in~\cite{cama19a}. 
Therein an excellent comprehensive review of previous results related to 
the symmetry analysis of the latter model was presented. 
The two-dimensional shallow water equations with constant Coriolis force were investigated in~\cite{ches09Ay},
where Lie symmetries were used to find a transformation relating this case to the shallow water equations in a resting reference frame.
This result was generalized in~\cite{ches11a}
via finding a point transformation mapping the shallow water equations over a constantly rotating parabolic basin
to the shallow water equations in a resting reference frame over a flat bottom topography.

Non-canonical Hamiltonian structures and generalized Hamiltonian structures of the shallow water equations
were considered in~\cite{salm07a}, see also~\cite{salm98a}.
In~\cite{salm07a}, these Hamiltonian structures were used to construct conservative numerical schemes for the shallow water equations,
a subject which was continued in~\cite{wan13Ay}, where conservation law characteristics were used for the same purpose.
Numerical schemes preserving Lie symmetries of the shallow water equations in Lagrangian and Eulerian coordinates
were considered in~\cite{bihl12By}.

In the present paper, we carry out the complete group classification of the class of systems
of two-dimensional shallow water equations with variable bottom topography,
which are of the form
\begin{equation}\label{eq:2DSWEs}
\begin{split}
&u_t+uu_x+vu_y+h_x=b_x,\\
&v_t+uv_x+vv_y+h_y=b_y,\\
&h_t+(uh)_x+(vh)_y=0.
\end{split}
\end{equation}
Here
$(u,v)$ is the horizontal fluid velocity averaged over the height of the fluid column,
$h$ is the thickness of a fluid column,
$b=b(x,y)$ is a parameter function that is the bottom topography measured downward with respect to a fixed reference level, 
and the gravitational acceleration is set to be equal 1 in dimensionless units.
In this class,
$(t,x,y)$ is the tuple of the independent variables,
$(u,v,h)$ is the tuple of the dependent variables and
$b$ is considered to be the arbitrary element of the class.
These quantities are graphically represented in Figure~\ref{picture}.
\begin{figure}
\centering
\includegraphics[width=.8\linewidth]{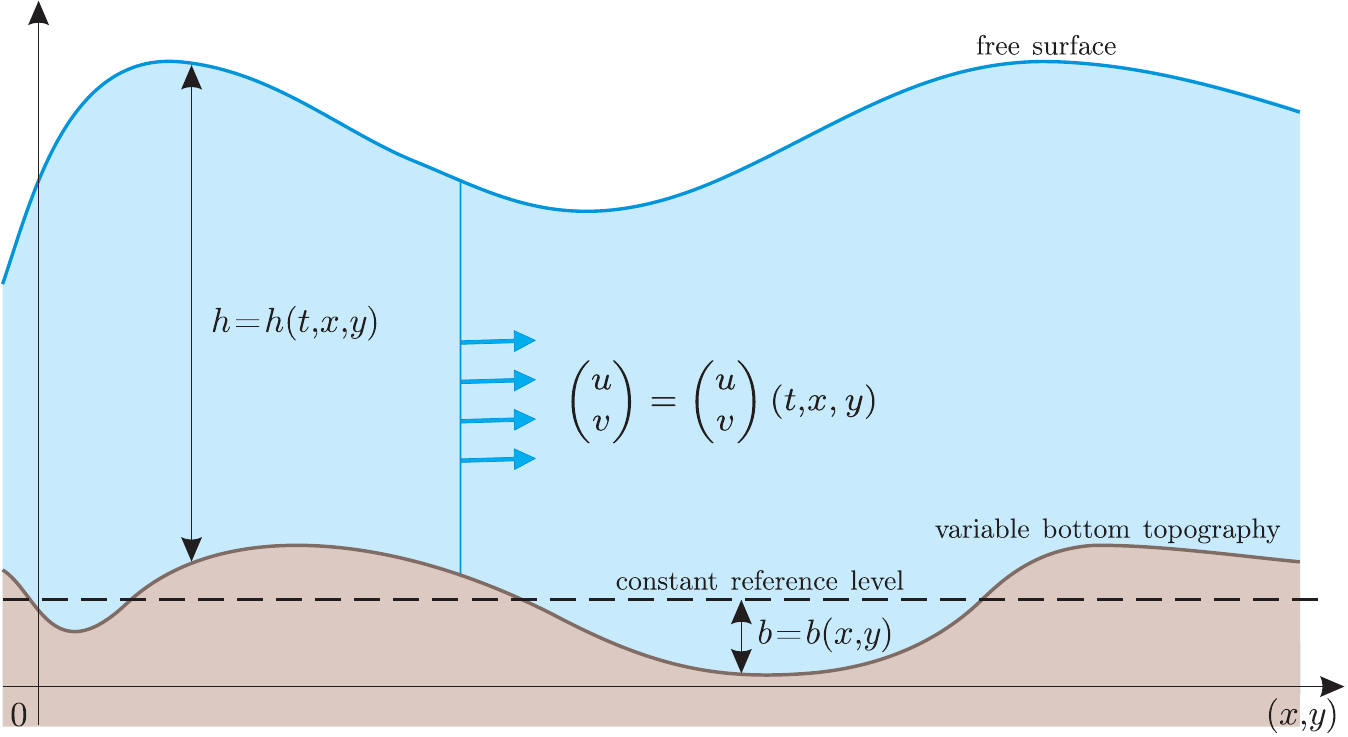}
\caption{The shallow water model.}
\label{picture}
\end{figure}

We classify cases of Lie-symmetry extensions for systems from the class~\eqref{eq:2DSWEs}
up to equivalence generated by the equivalence group~$G^\sim$ of this class.
Then we find additional equivalences among listed $G^\sim$-inequivalent cases of Lie-symmetry extensions.
These additional equivalences are induced by admissible point transformations within the class~\eqref{eq:2DSWEs}
that are not generated jointly by elements of~$G^\sim$ and by point symmetry groups of systems from the class~\eqref{eq:2DSWEs}.

We solve the group classification problem within the framework of the infinitesimal approach
using an optimized version of the \emph{method of furcate splitting}.
This method was suggested in~\cite{niki01a} in the course of group classification
of the class of nonlinear Schr\"odinger equations of the form
${\rm i}\psi_t+\triangle\psi+F(\psi,\psi^*)=0$
with an arbitrary number~$n$ of space variables.
Here, $\psi$ is an unknown complex-valued function of real variables~$(t,x_1,\dots,x_n)$,
and $F$ is an arbitrary sufficiently smooth function of $(\psi,\psi^*)$,
which is the arbitrary element of this class.
Subsequently, the method of furcate splitting was applied
to the group classification of various classes of (1+1)-dimensional variable-coefficient
reaction--convection--diffusion equations,
where arbitrary elements depend on single but possibly different arguments
\cite{ivan10a,ivan06a,opan18b,popo04a,vane12a,vane15c}.
Other classes whose arbitrary elements depend on single arguments were also classified, including 
higher-order Burgers-like equations~\cite{boyk01a},
Gardner equations with time-dependent coefficients~\cite{vane15b} 
and multidimensional nonlinear wave equations~\cite{vasi01a}.
Therefore, the present paper gives only the second example
of solving the group classification problem with the method of furcate splitting
for a class of (systems of) differential equations
with arbitrary elements depending on two arguments.
In the course of applying the method of furcate splitting,
we obtained a set of template-form equations,
which are inhomogeneous first-order quasilinear partial differential equations with respect to the arbitrary element~$b$
with two independent variables~$x$ and~$y$.
Each of these equations is canonically associated with a vector field
in the space with the coordinates $(x,y,b)$.
Optimizing the computation within the method of furcate splitting,
we show that the set of such vector fields is a Lie algebra
with respect to the Lie bracket of vector fields.

The further organization of the paper is as follows.
The equivalence group~$G^\sim$ of the class~\eqref{eq:2DSWEs} is computed in Section~\ref{sec:EquivalenceGroupSWE} 
by the algebraic method suggested in~\cite{hydo1998a,hydo00By}.
Section~\ref{sec:PreliminaryAnalysisSWE} contains the preliminary analysis of determining equations 
for Lie symmetries of systems from the class~\eqref{eq:2DSWEs} 
and the statement of results of the group classification of the class~\eqref{eq:2DSWEs} up to $G^\sim$-equivalence.
The proof of this classification is presented in Section~\ref{sec:ProofClassification}.
As the class~\eqref{eq:2DSWEs} is not semi-normalized, 
additional equivalences between $G^\sim$-equivalent cases of Lie-symmetry extensions have to be studied. 
This is done in Section~\ref{sec:AdditionalEquivalencesSWE} 
via comparing the structure of the corresponding maximal Lie invariance algebras.
In the final Section~\ref{sec:ConclusionsSWE} we summarize the findings of the paper and discuss possible future research directions.

\section{Equivalence group}\label{sec:EquivalenceGroupSWE}

According to the interpretation of $b$ as a varying arbitrary element or as a fixed function,
we will refer to~\eqref{eq:2DSWEs} as to a class of systems of differential equations or to a fixed system.
The complete auxiliary system for the arbitrary element~$b$ of the class~\eqref{eq:2DSWEs} consists of the equations
\begin{gather*}
b_u=b_{u_t}=b_{u_x}=b_{u_y}=0, \quad
b_v=b_{v_t}=b_{v_x}=b_{v_y}=0, \\
b_h=b_{h_t}=b_{h_x}=b_{h_y}=0, \quad
b_t=0. 
\end{gather*}
Note that there are no auxiliary inequalities for the arbitrary element~$b$.

The arbitrary element~$b$ depends only on independent variables.
Therefore, we can treat it as one more dependent variable and consider the extended system
\begin{equation}\label{2dSWE}
\begin{split}
&u_t+uu_x+vu_y+h_x=b_x,\\
&v_t+uv_x+vv_y+h_y=b_y,\\
&h_t+(uh)_x+(vh)_y=0,\\
&b_t=0.
\end{split}
\end{equation}
Here we also use the fact that the arbitrary element~$b$ does not depend on $t$ as well.

Since the arbitrary element~$b$ does not involve derivatives of dependent variables,
the generalized equivalence group~$G^\sim$ of the class~\eqref{eq:2DSWEs}
can be assumed to act in the space with the coordinates $(t,x,y,u,v,h,b)$
and thus to coincide with the point symmetry group~$G$ of the system~\eqref{2dSWE}.
Analogously, the generalized equivalence algebra~$\mathfrak g^{\sim}$ of the class~\eqref{eq:2DSWEs}
can be identified with the maximal Lie invariance algebra~$\mathfrak g$ of the system~\eqref{2dSWE}.
This is why it suffices to find~$\mathfrak g$ and~$G$ instead of~$\mathfrak g^{\sim}$ and~$G^\sim$, respectively.

To construct the group~$G$, we invoke Hydon's automorphism-based version~\cite{hydo1998a,hydo00By} of the algebraic method 
for finding discrete symmetries of systems of differential equations.%
\footnote{%
See also \cite{bihl2015a,bihl11Cy,card12Ay,card2015a,kont2017a} for further development, other versions and extensions of this method.
}
For this, we first need to compute the algebra~$\mathfrak g$,
and the infinitesimal method~\cite{blum89Ay,blum10Ay,olve86Ay,ovsi82A} is relevant here, see~\cite{akha91a}.
The algebra $\mathfrak g$ consists of the infinitesimal generators of one-parameter point symmetry groups of the system~\eqref{2dSWE},
which are vector fields in the space with coordinates $(t,x,y,u,v,h,b)$,
\begin{gather*}
\mathbf v=\tau\p_t+\xi^1\p_x+\xi^2\p_y+\eta^1\p_u+\eta^2\p_v+\eta^3\p_h+\eta^4\p_b,
\end{gather*}
where the components $\tau$, $\xi^1$, $\xi^2$ and $\eta^i$, $i=1,2,3,4$, are smooth functions of these coordinates.
For convenience, hereafter we simultaneously use the notation $(w^1,w^2,w^3,w^4)$ for $(u,v,h,b)$.
The infinitesimal invariance criterion implies that
\begin{equation}\label{cr_11}
\begin{split}
&\mathop{\rm pr}\nolimits^{(1)}\mathbf v(w^1_t+w^1w^1_x+w^2w^1_y+w^3_x-w^4_x)=0,\\
&\mathop{\rm pr}\nolimits^{(1)}\mathbf v(w^2_t+w^1w^2_x+w^2w^2_y+w^3_y-w^4_y)=0,\\
&\mathop{\rm pr}\nolimits^{(1)}\mathbf v(w^3_t+(w^1w^3)_x+(w^2w^3)_y)=0,\\
&\mathop{\rm pr}\nolimits^{(1)}\mathbf v(w^4_t)=0,
\end{split}
\end{equation}
whenever the system~\eqref{2dSWE} holds.
Here $\mathop{\rm pr}\nolimits^{(1)}\mathbf v$ is the first order prolongation of the vector field $\mathbf v$,
\begin{gather*}
\mathop{\rm pr}\nolimits^{(1)}\mathbf v=\mathbf v+\sum_{i=1}^{4}(\eta^{it}\p_{w^i_t}+\eta^{ix}\p_{w^i_x}+\eta^{iy}\p_{w^i_y})
\end{gather*}
with
$
\eta^{it}=\mathrm D_t(\eta^i-\tau w^i_t-\xi^1 w^i_x-\xi^2 w^i_y)+\tau w^i_{tt}+\xi^1 w^i_{tx}+\xi^2 w^i_{ty},
$
and similarly for $\eta^{ix}$ and $\eta^{iy}$; 
$\mathrm D_t$,~$\mathrm D_x$ and $\mathrm D_y$ denote the total derivative operators with respect to~$t$, $x$ and~$y$, 
respectively. 

We substitute the expressions for $w^i_t$, $i=1,\dots 4,$ in view of the system~\eqref{2dSWE} into the expanded equations~\eqref{cr_11} 
and then split them with respect to the derivatives $w^i_x$ and $w^i_y$, $i=1,\dots,4$. 
This procedure results in the system of differential equations on the components $\tau$, $\xi^1$, $\xi^2$ and $\eta^i$, $i=1,\dots,4$, 
of the vector field $\mathbf v$, which are called the determining equations. 
Integrating this system, we derive the explicit form of the vector field components,
\begin{gather*}
\tau=(c_5-c_7)t+c_1,\quad \xi^1=c_5x+c_6y+c_2, \quad \xi^2=-c_6x+c_5y+c_3,\\
\eta^1=c_7u+c_6v,\quad \eta^2=-c_6u+c_7v, \quad \eta^3=2c_7h,\quad \eta^4=2c_7b+c_4,
\end{gather*}
where $c_1,\dots ,c_7$ are arbitrary real constants.

Thus, the maximal Lie invariance algebra $\mathfrak g$ of the system~\eqref{2dSWE} is spanned by the seven vector fields%
\footnote{%
The components of vector fields from $\mathfrak g$
that correspond to the independent variables $(t,x,y)$ and dependent variables $(u,v,h)$ of the system~\eqref{eq:2DSWEs}
do not depend on the arbitrary element $b$.
Interpreting this result in terms of equivalence algebras,
we obtain that the generalized equivalence algebra~$\mathfrak g^{\sim}$ of the class~\eqref{eq:2DSWEs}
coincides with its usual equivalence algebra.
}
\begin{gather*}
P^t=\p_t,\quad P^x=\p_x, \quad P^y=\p_y, \quad P^b=\p_b,\quad
D^1=x\p_x+y\p_y+u\p_u+v\p_v+2h\p_h+2b\p_b,\\
D^2=t\p_t-u\p_u-v\p_v-2h\p_h-2b\p_b,\quad
J=x\p_y-y\p_x+u\p_v-v\p_u.
\end{gather*}
Let us fix the basis $\mathcal B=(P^t,P^x,P^y,P^b,D^1,D^2,J)$ of the Lie algebra~$\mathfrak g$.
Up to anticommutativity of the Lie bracket of vector fields,
the only nonzero commutation relations between the basis elements are
\begin{gather*}
[P^x,D^1]=P^x,\quad
[P^y,D^1]=P^y,\quad
[P^b,D^1]=2P^b,\\
[P^t,D^2]=P^t,\quad
[P^b,D^2]=-2P^b,\quad
[P^x,J]=P^y,\quad
[P^y,J]=-P^x.
\end{gather*}
In other words, the complete list of nonzero structure constants of the Lie algebra~$\mathfrak g$
in the basis~$\mathcal B$ is exhausted, up to permutation of subscripts, by
\begin{gather*}
c^2_{25}=1,\quad
c^3_{35}=1,\quad
c^4_{45}=2,\quad
c^1_{16}=1,\quad
c^4_{46}=-2,\quad
c^3_{27}=1,\quad
c^2_{37}=-1.
\end{gather*}
The general form $A=(a^i_j)_{i,j=1}^7$ of automorphism matrices of the algebra $\mathfrak g$ 
in the basis $\mathcal B$ can be found via solving the system of algebraic equations
\begin{equation}\label{sys_aut}
c^{k'}_{i'j'}a^{i'}_{i}a^{j'}_{j}=c^{k}_{ij}a^{k'}_{k},\quad i,j=1,\dots ,7,
\end{equation}
under the condition $\det A\ne0$.
Here we assume summation over the repeated indices.
As a result, we obtain that the automorphism group $\mathrm{Aut}(\mathfrak g)$ of $\mathfrak g$
can be identified with the matrix group that consists of the matrices of the general form
\begin{gather*}
A=\begin{pmatrix}
        a^1_1 & 0 & 0 & 0 & 0 & a^1_6 & 0\\
        0 & a^2_2 & -\varepsilon a^3_2 & 0 & a^2_5 & 0 & a^2_7  \\
        0 & a^3_2 & \varepsilon a^2_2 & 0 & -\varepsilon a^2_7 & 0 & \varepsilon a^2_5 \\
        0 & 0 & 0 & a^4_4 & -a^4_6 & a^4_6 & 0 \\
        0 & 0 & 0 & 0 & 1 & 0 & 0 \\
        0 & 0 & 0 & 0 & 0 & 1 & 0 \\
        0 & 0 & 0 & 0 & 0 & 0 & \varepsilon
  \end{pmatrix},
\end{gather*}
where $\varepsilon=\pm1$, and the remaining parameters $a^i_j$'s are arbitrary real constants with
\[
a^1_1\big((a^2_2)^2+(a^3_2)^2\big)a^4_4\ne0.
\]

\begin{theorem}\label{dsym}
A complete list of discrete symmetry transformations of the extended system~\eqref{2dSWE} 
that are independent up to combining with each other and with continuous symmetry transformations of this system 
is exhausted by two transformations alternating signs of variables,
\begin{gather}\label{eq:2DSWEsDiscreteEquivTrans}
\begin{split}
&(t,x,y,u,v,h,b)\mapsto (-t,x,y,-u,-v,h,b), \\
&(t,x,y,u,v,h,b)\mapsto (t,x,-y,u,-v,h,b).
\end{split}
\end{gather}
\end{theorem}

\begin{proof}
The maximal Lie invariance algebra $\mathfrak g$ of the system~\eqref{2dSWE} is finite-dimensional and nontrivial.
The complete automorphism group $\mathrm{Aut}(\mathfrak g)$ of $\mathfrak g$ is computed above.
It is not much wider than the inner automorphism group $\mathrm{Inn}(\mathfrak g)$ of $\mathfrak g$,
which consists of the linear operators on $\mathfrak g$ with matrices of the form
\begin{gather*}
\begin{pmatrix}
        e^{-\theta_6} & 0 & 0 & 0 & 0 & \theta_1 & 0 \\
        0 &  e^{-\theta_5}\cos\theta_7 & e^{-\theta_5}\sin\theta_7 & 0 & \theta_2 & 0 & -\theta_3 \\
        0 & -e^{-\theta_5}\sin\theta_7 & e^{-\theta_5}\cos\theta_7 & 0 & \theta_3 & 0 &  \theta_2 \\
        0 & 0 & 0 & e^{2\theta_6-2\theta_5} & 2\theta_4 & -2\theta_4 & 0 \\
        0 & 0 & 0 & 0 & 1 & 0 & 0 \\
        0 & 0 & 0 & 0 & 0 & 1 & 0 \\
        0 & 0 & 0 & 0 & 0 & 0 & 1
  \end{pmatrix},
\end{gather*}
where the parameters $\theta_1,$ $\dots$, $\theta_7$ are arbitrary constants.
Continuous point symmetries of the system~\eqref{2dSWE} can be easily found
by composing elements of one-parameter groups generated by basis elements of~$\mathfrak g$.
Moreover, such symmetries constitute the connected component of the identity transformation in the group $G$,
which induces the entire group $\mathrm{Inn}(\mathfrak g)$.
This is why it suffices to look only for discrete symmetry transformations,
and in the course of the related computation within the framework of the algebraic method, one can factor out inner automorphisms.
The quotient group $\mathrm{Aut}(\mathfrak g)/\mathrm{Inn}(\mathfrak g)$ can be identified 
with the matrix group consisting of the diagonal matrices of the form $\mathrm{diag}(\varepsilon',1,\varepsilon,\varepsilon'',1,1,\varepsilon)$,
where $\varepsilon,\varepsilon',\varepsilon''=\pm1$.
Suppose that the push-forward $\mathcal T_*$ of vector fields in the space with the coordinates $(t,x,y,u,v,h,b)$ by a point transformation
\begin{gather*}
\mathcal T\colon \quad (\tilde t,\tilde x,\tilde y,\tilde u,\tilde v,\tilde h,\tilde b)=(T,X,Y,U,V,H,B)(t,x,y,u,v,h,b)
\end{gather*}
generates the automorphism of $\mathfrak g$ with the matrix $\mathrm{diag}(\varepsilon',1,\varepsilon,\varepsilon'',1,1,\varepsilon)$, i.e.,
\begin{gather*}
\mathcal T_*P^t=\varepsilon' \tilde P^t, \quad \mathcal T_*P^x=\tilde P^x, \quad \mathcal T_*P^y=\varepsilon \tilde P^y, \quad
\mathcal T_*P^b=\varepsilon'' \tilde P^b, \\
\mathcal T_*D^1=\tilde D^1, \quad \mathcal T_*D^2= \tilde D^2, \quad \mathcal T_*J=\varepsilon \tilde J.
\end{gather*}
Here tildes over vector fields mean that these vector fields are given in the new coordinates.
The above conditions for $\mathcal T_*$ imply a system of differential equations for the components of $\mathcal T$,
\begin{gather*}
T_t=\varepsilon',\quad T=tT_t, \quad T_x=T_y=T_u=T_v=T_h=T_b=0,\\
X_x=1,\quad X=xX_x,\quad X_t=X_y=X_u=X_v=X_h=X_b=0,\\
Y_y=\varepsilon,\quad Y=yY_y,\quad Y_t=Y_x=Y_u=Y_v=Y_h=Y_b=0,\\
U_t=U_x=U_y=U_b=0,\quad V_t=V_x=V_y=V_b=0,\\
vU_u-uU_v=\varepsilon V, \quad uU_u+vU_v+2hU_h=U,\\
vV_u-uV_v=-\varepsilon U, \quad uV_u+vV_v+2hV_h=V,\\
H_t=H_x=H_y=H_b=0, \quad vH_u-uH_v=0, \quad uH_u+vH_v+2hH_h=2H,\\
B_b=\varepsilon'', \quad B_t=B_x=B_y=0,\quad vB_u-uB_v=0, \\
uB_u+vB_v+2hB_h=2B-2\varepsilon'' b.
\end{gather*}

{\samepage\noindent
The general solution of the system is
\begin{gather*}
T=\varepsilon' t,\quad X=x,\quad Y=\varepsilon y, \\[.5ex]
U=uF_1\left(\frac h{u^2+v^2}\right)+\varepsilon vF_2\left(\frac h{u^2+v^2}\right),\\[.5ex]
V=-uF_2\left(\frac h{u^2+v^2}\right)+\varepsilon vF_1\left(\frac h{u^2+v^2}\right),\\[.5ex]
H=(u^2+v^2)F_3\left(\frac h{u^2+v^2}\right),\quad B=\varepsilon''b+(u^2+v^2)F_4\left(\frac h{u^2+v^2}\right),
\end{gather*}
where $F_1$, $F_2$, $F_3$ and $F_4$ are arbitrary smooth functions of $h/(u^2+v^2)$.

}

We continue the computations within the framework of the direct method in order to complete the system of constraints for $\mathcal T$. 
Using the chain rule, we express all required transformed derivatives 
$\tilde{w}^i_{\tilde{t}}$, $\tilde{w}^i_{\tilde{x}}$, $\tilde{w}^i_{\tilde{y}}$, $i=1,\dots,4$, in terms of the initial coordinates. 
Then, we substitute the obtained expressions into the copy of the system~\eqref{2dSWE} in the new coordinates. 
The expanded system should identically be satisfied by each solution of the system~\eqref{2dSWE}.
This condition implies that
\begin{gather*}
T=\varepsilon' t,\quad X=x, \quad Y=\varepsilon y, \quad U=\varepsilon' u, \quad V=\varepsilon\varepsilon' v, \quad H=h,\quad B=b.
\end{gather*}
Therefore, discrete symmetries of the equation~\eqref{2dSWE} are exhausted, up to combining with continuous symmetries and with each other,
by the two involutions~\eqref{eq:2DSWEsDiscreteEquivTrans},
which are associated with the values $(\varepsilon',\varepsilon)=(-1,1)$ and $(\varepsilon', \varepsilon)=(1,-1)$, respectively.
\end{proof}

\begin{corollary}
The quotient group of the complete point symmetry group $G$ of the extended system~\eqref{2dSWE} 
with respect to its identity component is isomorphic to the group $\mathbb{Z}_2\times\mathbb{Z}_2$.
\end{corollary}

The complete point symmetry group $G$ of the extended system~\eqref{2dSWE} 
is generated by one-parameter point transformation groups associated with vector fields 
from the algebra $\mathfrak g$ and two discrete transformations given in Theorem~\ref{dsym}.

\begin{corollary}\label{2d_comp_sym}
The complete point symmetry group $G$ of the extended system~\eqref{2dSWE} consists of the transformations
\begin{gather}\label{eq:2DSWEsEquivTrans}
\begin{split}
&\tilde t=\delta_1t+\delta_2,\quad
 \tilde x=\delta_3x-\varepsilon\delta_4y+\delta_5, \quad
 \tilde y=\delta_4x+\varepsilon\delta_3y+\delta_6,\\
&\tilde u=\frac{\delta_3}{\delta_1}u-\varepsilon\frac{\delta_4}{\delta_1}v,\quad
 \tilde v=\frac{\delta_4}{\delta_1}u+\varepsilon\frac{\delta_3}{\delta_1}v,\quad
 \tilde h=\frac{\delta_3^{\,2}+\delta_4^{\,2}}{\delta_1^{\,2}}h,\quad
 \tilde b=\frac{\delta_3^{\,2}+\delta_4^{\,2}}{\delta_1^{\,2}}b+\delta_7,
\end{split}
\end{gather}
where $\varepsilon=\pm1$ and the parameters $\delta_i,\, i=1,\dots,7,$ are arbitrary constants with $\delta_1(\delta_3^{\ 2}+\delta_4^{\ 2})\ne0$.
\end{corollary}

Since the generalized equivalence group $G^\sim$ of the class of two-dimensional shallow water equations~\eqref{eq:2DSWEs} 
coincides with the complete point symmetry group $G$ of the system~\eqref{2dSWE}, 
we can rephrase Theorem~\ref{dsym} and Corollary~\ref{2d_comp_sym} in terms of equivalence transformations of the class~\eqref{eq:2DSWEs}.

\begin{theorem}\label{dequiv}
A complete list of discrete equivalence transformations of the class of two-dimen\-sional shallow water equations~\eqref{eq:2DSWEs} 
that are independent up to combining with each other and with continuous equivalence transformations of this class is exhausted
by two involutions~\eqref{eq:2DSWEsDiscreteEquivTrans} alternating signs of variables.
\end{theorem}

\begin{theorem}\label{2d_gen_equiv}
The generalized equivalence group~$G^\sim$ of the class of two-dimensional systems of shallow water equations~\eqref{eq:2DSWEs} 
coincides with the usual equivalence group of this class and consists of the transformations of the form~\eqref{eq:2DSWEsEquivTrans}.
\noprint{
\begin{gather*}
\tilde t=\delta_1t+\delta_2,\quad
\tilde x=\delta_3x-\varepsilon\delta_4y+\delta_5, \quad
\tilde y=\delta_4x+\varepsilon\delta_3y+\delta_6,\\
\tilde u=\frac{\delta_3}{\delta_1}u-\varepsilon\frac{\delta_4}{\delta_1}v,\quad
\tilde v=\frac{\delta_4}{\delta_1}u+\varepsilon\frac{\delta_3}{\delta_1}v,\quad
\tilde h=\frac{\delta_3^{\,2}+\delta_4^{\,2}}{\delta_1^{\,2}}h,\quad
\tilde b=\frac{\delta_3^{\,2}+\delta_4^{\,2}}{\delta_1^{\,2}}b+\delta_7,
\end{gather*}
where $\varepsilon=\pm1$ and the parameters $\delta_i,\, i=1,\dots,7,$ are arbitrary constants with $\delta_1(\delta_3^{\ 2}+\delta_4^{\ 2})\ne0$.
}
\end{theorem}

\section{Preliminary analysis and classification result}\label{result_section}\label{sec:PreliminaryAnalysisSWE}

Let $\mathcal L_b$ be a system from the class~\eqref{eq:2DSWEs} with a fixed value of the arbitrary element~$b$ 
and suppose that a vector field~$\mathbf v$ of the general form
\begin{gather*}
\begin{split}
\mathbf v={}&\tau(t,x,y,u,v,h)\p_t+\xi^1(t,x,y,u,v,h)\p_x+\xi^2(t,x,y,u,v,h)\p_y\\
          &{}+\eta^1(t,x,y,u,v,h)\p_u+\eta^2(t,x,y,u,v,h)\p_v+\eta^3(t,x,y,u,v,h)\p_h
\end{split}
\end{gather*}
defined in the space with the coordinates $(t,x,y,u,v,h)$ is 
the infinitesimal generator of a one-parameter Lie symmetry group for the system $\mathcal L_b$.
The set of such vector fields is the maximal Lie invariance algebra~$\mathfrak g_b$ of the system $\mathcal L_b$.

The infinitesimal invariance criterion requires that
\begin{gather}\label{infc}
\mathop{\rm pr}\nolimits^{(1)}\mathbf v(\mathcal L_b)\big|_{\mathcal L_b}=0.
\end{gather}
The first prolongation~$\mathop{\rm pr}\nolimits^{(1)}\mathbf v$ of the vector field~$\mathbf v$ is computed
similarly to the previous section.
We expand the condition~\eqref{infc} and confine it on the manifold defined by $\mathcal L_b$ in the corresponding first-order jet space,
assuming the first-order derivatives of the dependent variables $(u,v,h)$ with respect to $t$
as the leading ones and substituting for these derivatives in view of the system~$\mathcal L_b$,
\begin{gather*}
u_t=-uu_x-vu_y-h_x+b_x,\\
v_t=-uv_x-vv_y-h_y+b_y,\\
h_t=-(uh)_x-(vh)_y.
\end{gather*}
Then we split the obtained equations with respect to the first-order parametric derivatives,
which are the first-order derivatives of the dependent variables $(u,v,h)$ with respect to~$x$ and~$y$.
After an additional rearrangement and excluding equations that are differential consequences of the others,
we derive the system of determining equations for the components of the vector field~$\mathbf v$,
\begin{equation}\label{deteq2.2}
\begin{split}
&\tau_x=\tau_y=\tau_u=\tau_v=\tau_h=0,\\[.5ex]
&\xi^1_u=\xi^1_v=\xi^1_h=0,\quad \xi^2_u=\xi^2_v=\xi^2_h=0,\quad \xi^1_x=\xi^2_y,\quad \xi^1_y+\xi^2_x=0,\\[.5ex]
&\eta^1=(\xi^1_x-\tau_t)u+\xi^1_yv+\xi^1_t,\quad
 \eta^2=\xi^2_xu+(\xi^2_y-\tau_t)v+\xi^2_t,\quad
 \eta^3=2(\xi^1_x-\tau_t)h,\\[.5ex]
&\eta^1_t+u\eta^1_x+v\eta^1_y+\eta^3_x+(\eta^1_u-\tau_t)b_x+\eta^1_vb_y=\xi^1b_{xx}+\xi^2b_{xy},\\[.5ex]
&\eta^2_t+u\eta^2_x+v\eta^2_y+\eta^3_y+\eta^2_ub_x+(\eta^2_v-\tau_t)b_y=\xi^1b_{xy}+\xi^2b_{yy},\\[.5ex]
&\eta^3_t+u\eta^3_x+v\eta^3_y+h\eta^1_x+h\eta^2_y=0.
\end{split}
\end{equation}
Integrating the subsystem of the system~\eqref{deteq2.2} that consists of the equations not containing the arbitrary element~$b$,
we get the following form of the components of the vector field~$\mathbf v$:
\begin{equation}\label{dete2.1}
\begin{split}
&\tau=2F^1-c_1t,\\[.5ex]
&\xi^1=F^1_tx+F^0y+F^2,\\[.5ex]
&\xi^2=-F^0x+F^1_ty+F^3,\\[.5ex]
&\eta^1=(-F^1_t+c_1)u+F^0v+F^1_{tt}x+F^2_t,\\[.5ex]
&\eta^2=-F^0u+(-F^1_t+c_1)v+F^1_{tt}y+F^3_t,\\[.5ex]
&\eta^3=2(-F^1_t+c_1)h,
\end{split}
\end{equation}
where $F^i$, $i=1,2,3,4$, are sufficiently smooth functions of $t$, and $c_1$ is a constant.
From the last two equations of the system~\eqref{deteq2.2}, we derive as a differential consequence that $F^0_t=0$.
Thus, $F^0$~is a constant, and we will denote $c_2:=F^0$.
In other words, for any~$b$,
\[
\mathfrak g_b\subset\mathfrak g_\spanindex
:=\langle D(F^1),\,D^{\rm s},\,J,\,P(F^2,F^3)\rangle
 =\langle D(F^1),\,D^{\rm t},\,J,\,P(F^2,F^3)\rangle,
\]
where the parameters~$F^1$, $F^2$ and~$F^3$ run through the set of smooth functions of~$t$,
\begin{gather}\label{eq:2DSWEsRepresentationsForLieSymVFs}
\begin{split}
&D(F^1):=F^1\p_t+\tfrac12F^1_tx\p_x+\tfrac12F^1_ty\p_y-\tfrac12(F^1_tu-F^1_{tt}x)\p_u-\tfrac12(F^1_tv-F^1_{tt}y)\p_v-F^1_th\p_h,\!\\[.5ex]
&D^{\rm s}:=x\p_x+y\p_y+u\p_u+v\p_v+2h\p_h,\quad
 J:=x\p_y-y\p_x+u\p_v-v\p_u,\\[.5ex]
&P(F^2,F^3):=F^2\p_x+F^3\p_y+F^2_t\p_u+F^3_t\p_v,
\end{split}
\end{gather}
and it is convenient to denote 
$D^{\rm t}:=D(t)-\frac12D^{\rm s}=t\p_t-u\p_u-v\p_v-2h\p_h$ 
and sometimes use this vector field in the spanning set~\eqref{eq:2DSWEsRepresentationsForLieSymVFs} 
instead of~$D^{\rm s}$.

Up to antisymmetry, the nonzero commutation relations between the vector fields spanning~$\mathfrak g_\spanindex$ 
are exhausted by the following ones:
\begin{gather*}
[D(F^1),D(\tilde F^1)]=D(F^1\tilde F^1_t-\tilde F^1F^1_t),\\[.5ex]  
[D(F^1),P(F^2,F^3)]=P\big(F^1F^2_t-\tfrac12F^1_tF^2,F^1F^3_t-\tfrac12F^1_tF^3\big),\\[.5ex]
[D^{\rm s},P(F^2,F^3)]=-P(F^2,F^3),\quad 
[J,P(F^2,F^3)]=P(F^3,-F^2).
\end{gather*}
Therefore, the span~$\mathfrak g_\spanindex$ is an (infinite-dimensional) Lie algebra 
with respect to the Lie bracket of vector fields. 

The local one-parameter groups of point transformations generated by the vector fields~\eqref{eq:2DSWEsRepresentationsForLieSymVFs}
respectively consist of the following point transformations, where $\delta$ is the group parameter:
\begin{itemize}\itemsep=.5ex
\item
$\tilde t=T,\ \tilde x=T_t^{1/2}x,\ \tilde y=T_t^{1/2}y,\ \tilde u=T_t^{-1/2}u+\frac12T_{tt}T_t^{-3/2}x,\ \tilde v=T_t^{-1/2}v+\frac12T_{tt}T_t^{-3/2}y$, $\tilde h=T_t^{-1}h$, 
where $T=T(t,\delta):=\hat H\big(H(t)+\delta\big)$ with an antiderivative $H$ of~$1/F^1$ and the inverse $\hat H$ of~$H$ with respect to~$t$; 
\item
$\tilde t=t,\ \tilde x=e^\delta x,\ \tilde y=e^\delta y,\ \tilde u=e^\delta u,\ \tilde v=e^\delta v,\ \tilde h=e^{2\delta}h$;
\item
$\tilde t=t,\ \tilde x=x\cos\delta-y\sin\delta,\ \tilde y=x\sin\delta+y\cos\delta,\ \tilde u=u\cos\delta-v\sin\delta,\ \tilde v=u\sin\delta+v\cos\delta,$ $\tilde h=h$; 
\item
$\tilde t=t,\ \tilde x=x+\delta F^2(t),\ \tilde y=y+\delta F^3(t),\ \tilde u=u+\delta F^2_t(t),\ \tilde v=v+\delta F^3_t(t),\ \tilde h=h$.
\end{itemize}
These are 
arbitrary transformations of~$t$ with simultaneous linear transformations of the other variables with coefficients depending on~$t$, 
including shifts of~$t$ ($F^1=1$), concordant scalings of all variables ($F^1=t$) and time inversions ($F^1=t^2$);
scaling of the space variables~$(x,y)$ with simultaneous scalings of the dependent variables;
concordant rotations in the $(x,y)$- and $(u,v)$-planes; 
and generalized shifts of the space variables depending on~$t$, including their usual shifts ($F^2,F^3=\const$) and Galilean boosts ($F^2/t,F^3/t=\const$).

For elements of~$\mathfrak g_b$, which are of the form $2D(F^1)-c_1D^{\rm t}-c_2J+P(F^2,F^3)$,
the parameters~$F^1$, $F^2$, $F^3$, $c_1$ and~$c_2$ additionally satisfy two equations
implied by the last two equations from~\eqref{deteq2.2},
which explicitly involve the arbitrary element~$b$ and thus are the classifying equations for the class~\eqref{eq:2DSWEs}.
They can be integrated to the single equation
\begin{equation}\label{eq:2DSWEsClassifyingEq}
\begin{split}
&(F^1_tx+c_2y+F^2)b_x+(-c_2x+F^1_ty+F^3)b_y+2(F^1_t-c_1)b\\
&\qquad-F^1_{ttt}\frac{x^2+y^2}2-F^2_{tt}x-F^3_{tt}y-F^4=0,
\end{split}
\end{equation}
where $F^4$ is one more smooth parameter function of $t$.
The equation~\eqref{eq:2DSWEsClassifyingEq} can be considered as the only classifying equation instead of the above ones.
Thus, the group classification problem for the class~\eqref{eq:2DSWEs}
reduces to solving the equation~\eqref{eq:2DSWEsClassifyingEq} up to $G^\sim$-equivalence
with respect to the arbitrary element $b$ and the parameters~$F^1$, \dots, $F^4$, $c_1$ and~$c_2$.

The next theorem presents the results of the group classification of the class~\eqref{eq:2DSWEs} up to $G^\sim$-inequivalence.
In this theorem and in Section~\ref{sec:ProofClassification},
it is convenient to use, simultaneously with $(x,y)$, the polar coordinates~$(r,\varphi)$ on the $(x,y)$-plane,
\begin{gather*}
r:=\sqrt{x^2+y^2},\quad
\varphi:=\arctan\frac yx.
\end{gather*}

\begin{theorem}\label{thm:GroupClassificationOf2DSWEs1}
The kernel Lie invariance algebra of systems from the class~\eqref{eq:2DSWEs}
is $\mathfrak g^{\cap}=\langle D(1)\rangle$.
A~complete list of $G^\sim$-inequivalent Lie-symmetry extensions within the class~\eqref{eq:2DSWEs}
is exhausted by the following cases,
where $f$ denotes an arbitrary smooth function of a single argument,
$\alpha$,~$\beta$, $\mu$ and~$\nu$ are arbitrary constants with $\alpha\geqslant0\bmod G^\sim$, $\beta>0$
and additional constraints indicated in the corresponding cases,
$\varepsilon=\pm1\bmod G^\sim$ and $\delta\in\{0,1\}\bmod G^\sim$.
\end{theorem}

\begin{enumerate}
\item
\label{case1} $b=r^\nu f(\varphi+\alpha\ln r)$, \ $(\alpha,\nu)\ne(0,-2)$, \ $\nu\ne0$:\quad
$\mathfrak g_b=\big\langle D(1),\,4D(t)-(\nu+2)D^{\rm t}-2\alpha J\big\rangle$;
\item
\label{case2} $b=f(\varphi+\alpha\ln r)+\nu\ln r$, \ $\nu\in\{-1,0,1\}\bmod G^\sim$:\quad
$\mathfrak g_b=\big\langle D(1),\,2D(t)-D^{\rm t}-\alpha J\big\rangle$;
\item
\label{case3}$b=f(r)+\delta\varphi$:\quad
$\mathfrak g_b=\big\langle D(1),\,J\big\rangle$;
\item
\label{case4} $b=f(r)e^{\beta\varphi}$:\quad
$\mathfrak g_b=\big\langle D(1),\,2J-\beta D^{\rm t}\big\rangle$;
\item
\label{case5} $ b=f(y)e^x$:\quad
$\mathfrak g_b=\big\langle D(1),\,D^{\rm t}-P(2,0)\big\rangle$;
\item
\label{case6}
\begin{enumerate}
\item 
\label{case6a} $b=r^{-2}f(\varphi)$:\quad
$\mathfrak g_b=\big\langle D(1),\,D(t),\,D(t^2)\big\rangle$;
\item
\label{case6b} $b=r^{-2}f(\varphi)+\frac12r^2$:\quad
$\mathfrak g_b=\big\langle D(1),\,D(e^{2t}),\,D(e^{-2t})\big\rangle$;
\item
\label{case6c} $b=r^{-2}f(\varphi)-\frac12r^2$:\quad
$\mathfrak g_b=\big\langle D(1),\,D(\cos 2t),\,D(\sin 2t)\big\rangle$;
\end{enumerate}
\item
\label{case7} $b=f(y)+\delta x$:\quad
$\mathfrak g_b=\big\langle D(1),\,P(1,0),\,P(t,0)\big\rangle$;
\item
\label{case8} $b=f(y)+\tfrac12x^2$:\quad
$\mathfrak g_b=\big\langle D(1),\,P(e^t,0),\,P(e^{-t},0)\big\rangle$;
\item
\label{case9} $b=f(y)-\tfrac12x^2$:\quad
$\mathfrak g_b=\big\langle D(1),\,P(\cos t,0),\,P(\sin t,0)\big\rangle$;
\item
\label{case10} $b=\delta\varphi-\nu\ln r$, $\nu=\pm1\bmod G^\sim$ if $\delta=0$:\quad
$\mathfrak g_b=\big\langle D(1),\,2D(t)-D^{\rm t},\,J\big\rangle$;
\item
\label{case11} $b=\varepsilon r^{\nu}e^{\alpha\varphi}$,
$\nu\ne-2$, $(\alpha,\nu)\notin\{(0,0),(0,2)\}$:\ \
$\mathfrak g_b=\big\langle D(1),\,4D(t)-(\nu+2)D^{\rm t},\,2J-\alpha D^{\rm t}\big\rangle$;
\item\label{case12} 
\begin{enumerate}
\item
\label{case12a} $b=\varepsilon r^{-2}e^{\alpha\varphi}$:\quad
$\mathfrak g_b=\big\langle D(1),\,D(t),\,D(t^2),\,\alpha D^{\rm s}+4J\big\rangle$;
\item
\label{case12b} $b=\varepsilon r^{-2}e^{\alpha\varphi}+\frac12r^2$:\quad
$\mathfrak g_b=\big\langle D(1),\,D(e^{2t}),\,D(e^{-2t}),\,\alpha D^{\rm s}+4J\big\rangle$;
\item
\label{case12c} $b=\varepsilon r^{-2}e^{\alpha\varphi}-\frac12r^2$:\quad
$\mathfrak g_b=\big\langle D(1),\,D(\cos 2t),\,D(\sin 2t),\,\alpha D^{\rm s}+4J\big\rangle$;
\end{enumerate}
\item
\label{case13} $b=\varepsilon|y|^{\nu}+\delta x$, $\nu\notin\{-2,0,2\}$:\\
$\mathfrak g_b=\big\langle D(1),\,4D(t)-(\nu+2)D^{\rm t}-\delta(\nu-1)P(t^2,0),\,P(1,0),\,P(t,0)\big\rangle$;
\item
\label{case14} $b=\varepsilon\ln|y|+\delta x$:\quad
$\mathfrak g_b=\big\langle D(1),\,2D(t)-D^{\rm t}-\frac12\delta P(t^2,0)),\,P(1,0),\,P(t,0)\big\rangle;$
\item
\label{case15} $b=\varepsilon e^y+\delta x$:\quad
$\mathfrak g_b=\big\langle D(1),\,D^{\rm t}-P(\delta t^2,2),\,P(1,0),\,P(t,0)\big\rangle;$
\item
\begin{enumerate}
\item
\label{case16a} $b=\varepsilon y^{-2}+\delta x$:\quad
$\mathfrak g_b=\big\langle D(1),\,D(t)+\frac34\delta P(t^2,0),\,D(t^2)+\frac12\delta P(t^3,0),\,P(1,0),\,P(t,0)\big\rangle;$
\item
\label{case16b} $b=\varepsilon y^{-2}+\frac12r^2$:\quad
$\mathfrak g_b=\big\langle D(1),\,D(e^{2t}),\,D(e^{-2t}),\,P(e^t,0),\,P(e^{-t},0)\big\rangle$;
\item
\label{case16c} $b=\varepsilon y^{-2}-\frac12r^2$:\quad
$\mathfrak g_b=\big\langle D(1),\,D(\cos 2t),\,D(\sin 2t),\,P(\cos t,0),\,P(\sin t,0)\big\rangle$;
\end{enumerate}
\item
\label{case17} $b=\tfrac12x^2+\tfrac12\beta^2y^2$, $0<\beta<1$:\quad
$\mathfrak g_b=\big\langle D(1),\,D^{\rm s},\,P(e^t,0),\,P(e^{-t},0),\,P(0,e^{\beta t}),\,P(0,e^{-\beta t})\big\rangle$;
\item
\label{case18} $b=\frac12x^2+\delta y$:\quad
$\mathfrak g_b=\big\langle D(1),\,D^{\rm s}-\tfrac12\delta P(0,t^2),\,P(e^t,0),\,P(e^{-t},0),\,P(0,1),\,P(0,t)\big\rangle$;
\item
\label{case19} $b=\tfrac12x^2-\tfrac12\beta^2y^2$, $\beta>0$:\quad
$\mathfrak g_b=\big\langle D(1),\,D^{\rm s},\,P(e^t,0),\,P(e^{-t},0),\,P(0,\cos\beta t),\,P(0,\sin\beta t)\big\rangle$;
\item
\label{case20} $b=-\frac12x^2+\delta y$:\quad
$\mathfrak g_b=\big\langle D(1),\,D^{\rm s}-\tfrac12\delta P(0,t^2),\,P(\cos t,0),\,P(\sin t,0),\,P(0,1),\,P(0,t)\big\rangle$;
\item
\label{case21} $b=-\tfrac12x^2-\tfrac12\beta^2y^2$, $0<\beta<1$:\\[2pt]
$\mathfrak g_b=\big\langle D(1),\,D^{\rm s},\,P(\cos t,0),\,P(\sin t,0),\,P(0,\cos\beta t),\,P(0,\sin\beta t)\big\rangle$;
\item
\begin{enumerate}
\item
\label{case22a} $b=0$:\quad
$\mathfrak g_b=\big\langle D(1),\,D(t),\,D(t^2),\,D^{\rm s},\,J,\,P(1,0),\,P(t,0),\,P(0,1),\,P(0,t)\big\rangle$;
\item
\label{case22b} $b=x$:\quad
$\mathfrak g_b=\big\langle D(1),\,D(t)+\tfrac34P(t^2,0),\,D(t^2)+\tfrac12P(t^3,0),\,D^{\rm s}-\tfrac12P(t^2,0),\,J-\tfrac12P(0,t^2)$,\\[.5ex]
$\phantom{\mbox{$b=x$:\quad}\mathfrak g_b=\big\langle}P(1,0),\,P(t,0),\,P(0,1),\,P(0,t)\big\rangle$;
\item
\label{case22c} $b=\frac12r^2$:\ \
$\mathfrak g_b=\big\langle D(1),\,D(e^{2t}),\,D(e^{-2t}),\,D^{\rm s},\,J,\,P(e^t,0),\,P(e^{-t},0),\,P(0,e^t),\,P(0,e^{-t})\big\rangle$;\!
\item
\label{case22d} $b=-\frac12r^2$:\quad
$\mathfrak g_b=\big\langle D(1),\,D(\cos2t),\,D(\sin2t),\,D^{\rm s},\,J$,\\[.5ex]
$\phantom{\mbox{$b=-\frac12r^2$:\quad}\mathfrak g_b=\big\langle}P(\cos t,0),\,P(\sin t,0),\,P(0,\cos t),\,P(0,\sin t)\big\rangle$.
\end{enumerate}
\end{enumerate}

\begin{remark}\label{rem:2DSWEsOnMaxLieSymExtensions}
For Cases~\ref{case1}--\ref{case9} to really present maximal Lie-symmetry extensions,
the parameter function~$f$ should take only values
for which the corresponding values of the arbitrary element~$b$
are not $G^\sim$-equivalent to ones from the other listed cases.
\end{remark}

\begin{remark}\label{rem:2DSWEsOnEquivInGroupClassification}
The usage of the $G^\sim$-equivalence as the principal equivalence 
in the course of solving the group classification problem for the class~\eqref{eq:2DSWEs}
is natural from the physical standpoint.
The equivalence group~$G^\sim$ of the class~\eqref{eq:2DSWEs}, which is presented in Theorem~\ref{2d_gen_equiv}, 
is generated by elementary point transformations with the obvious physical interpretation 
as basic coordinate changes, more specifically,   
the shifts, the scalings and the reflections of the time and the space variables, the rotations of the space variables 
and the vertical shifts of the bottom topography. 
(The last transformations constitute a subgroup of~$G^\sim$, 
which is the gauge equivalence group  of the class~\eqref{eq:2DSWEs}, 
see~\cite{popo10Ay} for definitions.)
Roughly speaking, two bottom topographies are $G^\sim$-equivalent 
if they coincide up to simple coordinate changes. 
On the other hand, the $G^\sim$-equivalence is a necessary component of 
the rigorous formulation of the group classification problem for the class~\eqref{eq:2DSWEs} 
within the framework of the classical Lie--Ovsiannikov theory of group analysis of differential equations 
\cite[Chapter~III]{ovsi82A}, which is justified by the fact 
that a point transformation between two systems of differential equations 
induces an isomorphism between their maximal Lie invariance algebras.  
Moreover, the construction and the usage of the equivalence group of the class~\eqref{eq:2DSWEs} 
are much simpler than those of the corresponding equivalence groupoid; 
see Section~\ref{sec:AdditionalEquivalencesSWE} below. 
\end{remark}

\begin{remark}\label{rem:2DSWEsOnPhi} 
Among bottom topographies listed in Theorem~\ref{thm:GroupClassificationOf2DSWEs1}, 
there are a number of those that are not globally defined. 
This fact is a manifestation of the local (in several aspects) nature of point symmetries in general. 
It is especially evident for 
Cases~\ref{case1}--\ref{case4}, \ref{case6} and \ref{case10}--\ref{case12}, 
where the associated bottom topographies are represented in polar coordinates. 
For $2\pi$-periodic values of the parameter function~$f$ in Cases~\ref{case1}, \ref{case2} and~\ref{case6}, 
the corresponding bottom topographies have singularities only at the origin.
In Case~\ref{case3} with $\delta=0$, the domain of~$b$ with a fixed value of the parameter function~$f$ 
merely depends on properties of this value, including its domain and its behavior at the origin. 
At the same time, for other values of~$f$ and~$\delta$, and for the others of the above cases, 
the domains of bottom topographies are necessarily contained in the plane 
cut along the $x$-axis from the origin to~$+\infty$, 
and thus they cannot be preserved by transformations involving rotations 
even though such transformations are among Lie symmetries of these cases, except Case~\ref{case6}.%
\footnote{%
The arbitrary element~$b$ appears in the shallow water equations in the form of its gradient, 
and in contrast to~$b$ itself, the gradient can be defined in rotationally invariant domains 
in Cases~\ref{case3} and \ref{case10} with $\delta\ne0$.
}
Nevertheless, such symmetries can still be used 
for finding (at least local) exact solutions of the shallow water equations with the discussed bottom topographies. 
\end{remark}

\begin{remark}\label{rem:2DSWEsKnownResultsOnLieSyms}
The maximal Lie invariance algebra~$\mathfrak g_0$ of the system~\eqref{eq:2DSWEs} with $b=0$ (Case~\ref{case22a})
was given in \cite[Section~5.4]{ibra85A} 
via observing that this system coincides, after re\nobreakdash-interpreting the meaning of the dependent variables, 
with the system of equations for two-dimensional isotropic gas flows with adiabatic exponent two, 
and the maximal Lie invariance algebra of the latter system had been known due to Ovsiannikov \cite[Section~11]{ovsi82A}. 
Cases~\ref{case21} and~\ref{case22d} exhaust, up to $G^\sim$-equivalence, 
Lie-symmetry extensions for shallow liquids in elliptic paraboloidal basins, 
where $b=-c_1x^2-c_2y^2$ with positive constants~$c_1$ and~$c_2$. 
Similar Lie-symmetry extensions for the rotating reference frame were computed in~\cite{levi89By}. 
For the subclass of systems~$\mathcal L_b$ with rotationally invariant values of the arbitrary element~$b$, $b=b(r)$, 
a complete list of Lie-symmetry extensions inequivalent with respect to the equivalence group 
of this subclass is exhausted by Cases~\ref{case3}$_{\delta=0}$, \ref{case10}$_{\delta=0}$,
\ref{case11}$_{\alpha=0}$, \ref{case12}$_{\alpha=0}$, \ref{case22a}, \ref{case22b} and \ref{case22c}. 
Although this subclass was considered in~\cite{tito89a}, Case~\ref{case10}$_{\delta=0}$ was missed there.
\end{remark}

\begin{corollary}
The dimension of the maximal Lie invariance algebra of any system from the class~\eqref{eq:2DSWEs} is not greater than nine.
More specifically, $\dim \mathfrak g_b\in\{1,2,3,4,5,6,9\}$ for any $b=b(x,y)$.
We also have $\bigcup_b\mathfrak g_b\subsetneq\mathfrak g_\spanindex$.
\end{corollary}

\begin{corollary}
A system from the class~\eqref{eq:2DSWEs} is invariant with respect to a six-dimen\-sional Lie algebra
if and only if the corresponding value of the arbitrary element~$b$ is at most a quadratic polynomial in~$(x,y)$.
\end{corollary}

\section{Proof of the classification}\label{sec:ProofClassification}

According to the method of furcate splitting,
we fix an arbitrary value of the variable $t$ in the classifying equation~\eqref{eq:2DSWEsClassifyingEq}
and obtain the following template form of equations for the arbitrary element $b$:
\begin{equation}\label{eq:TemplateForm}
\begin{split}
&a_1(xb_x+yb_y)+a_2(yb_x-xb_y)+a_3b_x+a_4b_y+a_5b\\
&\qquad{}+a_6\frac{x^2+y^2}2+a_7x+a_8y+a_9=0,
\end{split}
\end{equation}
where $a_1,\dots, a_9$ are constants.
For each value of the arbitrary element~$b$,
we denote by $k=k(b)$ the maximal number of template-form equations
with linearly independent coefficient tuples $\bar a^i=(a^i_1,\dots, a^i_9)$, $i=1,\dots,k$,
that are satisfied by this value of~$b$.
It is obvious that $0\leqslant k\leqslant9$.
Moreover, if $k>0$, then for the system of template-form equations
\begin{equation}\label{eq:2DSWEsSysOfTemplateFormEqs}
\begin{split}
&a^i_1(xb_x+yb_y)+a^i_2(yb_x-xb_y)+a^i_3b_x+a^i_4b_y+a^i_5b\\
&\qquad{}+a^i_6\frac{x^2+y^2}2+a^i_7x+a^i_8y+a^i_9=0,\quad i=1,\dots,k,
\end{split}
\end{equation}
with $\rank A=k$ to be consistent with respect to~$b$,
it is required that $k\leqslant5$.
Here
\[
A:=(a^i_j)_{j=1,\dots,9}^{i=1,\dots,k}\,, \quad A_l:=(a^i_j)_{j=1,\dots,l}^{i=1,\dots,k}\,,\quad 1\leqslant l\leqslant9, 
\]
are the matrix of coefficients of the system~\eqref{eq:2DSWEsSysOfTemplateFormEqs} and
its submatrix constituted by the first $l$ columns of~$A$, respectively, and thus $A_9=A$.
We also have $\rank A_5=\rank A=k$, and, if $k<5$, $\rank A_4=k$ as well.
Indeed, if the last condition is not satisfied,
the system~\eqref{eq:2DSWEsSysOfTemplateFormEqs} has an algebraic consequence of the form
$b=R(x,y):=\beta_3(x^2+y^2)+\beta_{11}x+\beta_{12}y+\beta_0$,
where $\beta_0$, $\beta_{11}$, $\beta_{12}$ and~$\beta_3$ are constants,
and such values of $b$ satisfy five independent template-form equations (see the case $k=5$ below),
$xb_x+yb_y=xR_x+yR_y$, $yb_x-xb_y=yR_x-xR_y$, $b_x=R_x$, $b_y=R_y$, and $b=R$.

To check the consistency of the system~\eqref{eq:2DSWEsSysOfTemplateFormEqs} with $k>1$,
to the $i$th equation of this system for each $i=1,\dots,k$
we relate the vector field
\begin{equation}\label{eq:AssociatedVectorField}
\begin{split}
\mathbf v_i={}&\left(a^i_1x+a^i_2y+a^i_3\right)\p_x+\left(a^i_1y-a^i_2x+a^i_4\right)\p_y\\
&{}-\left(a^i_5b+\tfrac12a^i_6(x^2+y^2)+a^i_7x+a^i_8y+a^i_9\right)\p_b.
\end{split}
\end{equation}
Note that $\mathbf v_1,\dots,\mathbf v_k\in\mathfrak a$, where
\[
\mathfrak a:=\big\langle x\p_x+y\p_y,\,-x\p_y+y\p_x,\,\p_x,\,\p_y,\,b\p_b,\,(x^2+y^2)\p_b,\,x\p_b,\,y\p_b,\,\p_b\big\rangle.
\]
The span $\mathfrak a$ is closed with respect to the Lie bracket of vector fields,
i.e., it is a Lie algebra, and thus $[\mathbf v_i,\mathbf v_{i'}]\in\mathfrak a$, $i,i'=1,\dots,k$.
More specifically,
\begin{gather*}
[\mathbf v_i,\mathbf v_{i'}]\in[\mathfrak a,\mathfrak a]=\big\langle \p_x,\,\p_y,(x^2+y^2)\p_b,\,x\p_b,\,y\p_b,\,\p_b\big\rangle\subset\mathfrak a, \quad i,i'=1,\dots,k.
\end{gather*}
In other words, the equation on $b$ that is associated with $[\mathbf v_i,\mathbf v_{i'}]$
is a differential consequence of the system~\eqref{eq:2DSWEsSysOfTemplateFormEqs} 
and has the same template form~\eqref{eq:TemplateForm}.
By its definition, the number~$k=k(b)$ is equal to the maximal number of linearly independent vector fields
associated with template-form equations for the corresponding value of the arbitrary element~$b$.
Therefore,
\begin{gather}\label{eq:CompatibilityCondInTermsOfAssociatedVFs}
[\mathbf v_i,\mathbf v_{i'}]\in\big\langle\mathbf v_1,\dots,\mathbf v_k\big\rangle, \quad i,i'=1,\dots,k.
\end{gather}
We can also use the counterpart of the condition~\eqref{eq:CompatibilityCondInTermsOfAssociatedVFs}
for the projections $\hat{\mathbf v}_1$, \dots, $\hat{\mathbf v}_k$ of the vector fields $\mathbf v_1$, \dots, $\mathbf v_k$
to the space with the coordinates $(x,y)$,
\begin{gather}\label{eq:CompatibilityCondInTermsOfProjectionsOfAssociatedVFs}
[\hat{\mathbf v}_i,\hat{\mathbf v}_{i'}]\in\big\langle\hat{\mathbf v}_1,\dots,\hat{\mathbf v}_k\big\rangle, \quad i,i'=1,\dots,k.
\end{gather}

To simplify the computation, we can gauge coefficients of the system~\eqref{eq:2DSWEsSysOfTemplateFormEqs}
by linearly combining its equations and using transformations from $G^\sim$.
In particular, we can set $a^i_j=1$,
dividing the entire $i$th equation of~\eqref{eq:2DSWEsSysOfTemplateFormEqs} by~$a^i_j$ if $a^i_j\ne0$.
In the case $(a^i_1,a^i_2)\ne(0,0)$, we can make $a^i_3=a^i_4=0$
with point equivalence transformations of simultaneous shifts with respect to~$x$ and~$y$.
Another possibility is to use these shifts for setting $a^i_7=a^i_8=0$ if $a^i_6\ne0$.
Similarly, if $a^i_5\ne0$, then we can shift~$b$ to set $a^i_9=0$.

The case with $k=0$ corresponds to the kernel Lie invariance algebra~$\mathfrak g^\cap$
of systems from the class~\eqref{eq:2DSWEs},
which is also the Lie invariance algebra for a general value of~$b$.
For elements of~$\mathfrak g^\cap$,
the classifying equation~\eqref{eq:2DSWEsClassifyingEq} is identically satisfied by~$b$.
Thus, we can successively split it with respect to~$b$ and its derivatives
and with respect to~$x$ and~$y$ to obtain $F^1_t=F^2=F^3=0$ and $c_2=c_1=0$,
i.e.,
$\tau=\const$, $\xi^1=\xi^2=\eta^1=\eta^2=\eta^3=0$.
In other words,
the algebra~$\mathfrak g^\cap$ is one-dimensional and spanned by the only basis element $\p_t$,
\[
\mathfrak g^\cap=\langle\p_t\rangle.
\]

In the next sections, we separately consider the cases $k=1$, \dots, $k=5$.
For each of these cases,
we make the following steps, splitting the consideration into subcases
depending on values of the parameters~$a^i_1$, \dots, $a^i_9$, $i=1,\dots,k$:
\begin{itemize}\itemsep=0ex
\item
find the values of the parameters~$a$'s
for which the corresponding system~\eqref{eq:2DSWEsSysOfTemplateFormEqs}
is compatible and follow from the equation~\eqref{eq:2DSWEsClassifyingEq},
\item
gauge, if possible, some of the parameters~$a$'s
by recombining template-form equations and by transformations from the group~$G^\sim$
and re-denote the remaining parameters~$a$'s,
\item
integrate the system~\eqref{eq:2DSWEsSysOfTemplateFormEqs}
with respect to the arbitrary element~$b$,
\item
gauge, if possible, the integration constant by transformations from the group~$G^\sim$, and
\item
via solving the system of determining equations with respect to the parameters
$c_1$, $c_2$, $F^1$, $F^2$ and $F^3$,
construct the maximal Lie invariance algebras $\mathfrak g_b$
of systems from the class~\eqref{eq:2DSWEs} with the obtained values of~$b$.
\end{itemize}
The order of the steps can vary, and some steps can intertwine.

\subsection{One independent template-form equation}

In the case $k=1$, the right-hand side of the equation~\eqref{eq:2DSWEsClassifyingEq}
is proportional to the right-hand side of the single equation~\eqref{eq:2DSWEsSysOfTemplateFormEqs}
with the proportionality coefficient~$\lambda$ that is a sufficiently smooth function of~$t$,
\begin{gather*}
F^1_t(xb_x+yb_y)+c_2(yb_x-xb_y)+F^2b_x+F^3b_y+2(F^1_t-c_1)b\\
\qquad{}-F^1_{ttt}\frac{x^2+y^2}2-F^2_{tt}x-F^3_{tt}y-F^4\\
=\lambda\bigg(a^1_1(xb_x+yb_y)+a^1_2(yb_x-xb_y)+a^1_3b_x+a^1_4b_y+a^1_5b\\
\qquad{}+a^1_6\frac{x^2+y^2}2+a^1_7x+a^1_8y+a^1_9\bigg).
\end{gather*}
The function~$\lambda$ does not vanish for any vector field
from the complement of~$\mathfrak g^\cap$ in~$\mathfrak g_b$.
We can split the last equation with respect to derivatives of $b$, including~$b$ itself, and the independent variables~$x$ and~$y$.
As a result, we obtain the system
\begin{equation}\label{k1}
\begin{split}
&F^1_t=a^1_1\lambda,\quad c_2=a^1_2\lambda,\quad F^2=a^1_3\lambda,\quad F^3=a^1_4\lambda,\quad 2(F^1_t-c_1)=a^1_5\lambda,\\[1ex]
&F^1_{ttt}=-a^1_6\lambda,\quad F^2_{tt}=-a^1_7\lambda,\quad F^3_{tt}=-a^1_8\lambda,\quad F^4=-a^1_9\lambda.
\end{split}
\end{equation}
The condition $\rank A_4=k=1$ means here that $(a^1_1,a^1_2,a^1_3,a^1_4)\ne(0,0,0,0)$.
Therefore, the further consideration splits into three cases,
\[
a^1_1\ne0;\quad a^1_1=0,\ a^1_2\ne0;\quad a^1_1=a^1_2=0,\ (a^1_3,a^1_4)\ne(0,0).
\]

\noindent$\boldsymbol{a^1_1\ne0.}$
We can set $a^1_1=1$ by rescaling the entire equation~\eqref{eq:2DSWEsSysOfTemplateFormEqs}.
To simplify the computation we gauge other coefficients of~\eqref{eq:2DSWEsSysOfTemplateFormEqs} by transformations from $G^\sim$.
Thus, we can make $a^1_3=a^1_4=0$ with point equivalence transformations of simultaneous shifts with respect to $x$ and $y$.
Consequently, we have that $F^2=0$ and $F^3=0$.
Then the system~\eqref{k1} implies that $a^1_7=a^1_8=0$.
The equation~\eqref{eq:2DSWEsSysOfTemplateFormEqs} reduces in the polar coordinates~$(r,\varphi)$ to the form
\[
rb_r-a^1_2b_\varphi+a^1_5b+\frac12a^1_6r^2+a^1_9=0.
\]
The integration of the above equation depends on the values of the parameters $a^1_j$, $j=2,5,6,9$.

If $(a^1_2,a^1_5)=(0,2)$, then we can set $a^1_9=0$ by shifts with respect to $b$
and $a^1_6\in\{0,-4,4\}$ by scaling equivalence transformations,
which leads to Cases \ref{case6a}, \ref{case6b} and \ref{case6c}
of Theorem~\ref{thm:GroupClassificationOf2DSWEs1}, respectively.

If $(a^1_2,a^1_5)\ne(0,2)$, in view of the system~\eqref{k1}
we get that $c_1=(1-a^1_5/2)\lambda$, $c_2=a^1_2\lambda$.
Thus, $\lambda$ is a constant, which yields that $F^1_t$ is also constant, and therefore $a^1_6=0$.
Depending on whether $a^1_5\ne0$ or $a^1_5=0$, we get Cases~\ref{case1} and~\ref{case2}, respectively.
In the former case, we additionally set $a^1_9=0$ by shifts with respect to $b$.

\medskip\par\noindent$\boldsymbol{a^1_1=0,\, a^1_2\ne0.}$
Rescaling the equation~\eqref{eq:2DSWEsSysOfTemplateFormEqs} and using shifts with respect to $x$ and $y$,
we can set $a^1_2=1$ and $a^1_3=a^1_4=0$.
In view of the system~\eqref{k1}, the above conditions for~$a$'s imply
$F^1_t=0$, $F^2=F^3=0$, $\lambda=c_2$, and thus $a^1_6=a^1_7=a^1_8=0$.
In the polar coordinates~$(r,\varphi)$,
the equation~\eqref{eq:2DSWEsSysOfTemplateFormEqs} takes the form
$b_\varphi=a^1_5b+a^1_9$.
Integrating this equation separately for $a^1_5=0$ and for $a^1_5\ne0$
we get Cases~\ref{case3} and~\ref{case4}, respectively.
Under the former condition, we additionally set
$a^1_9=-1$ by equivalence transformations of scalings and alternating the signs of $(y,v)$ if $a^1_9\ne0$.
We also can set $a^1_9=0$ by shifts with respect to $b$ if $a^1_5\ne0$.

\par\medskip\par\noindent $\boldsymbol{a^1_1=a^1_2=0,\, (a^1_3,a^1_4)\ne(0,0).}$
Due to rotation equivalence transformations
and the possibility of scaling the entire equation~\eqref{eq:2DSWEsSysOfTemplateFormEqs},
we can rotate and scale the vector $(a^1_3,a^1_4)$ to set $a^1_3=1$ and $a^1_4=0$.
From the system~\eqref{k1}, we derive that $F^1_t=F^3=0$, $c_2=0$, $2c_1=-a^1_5\lambda$,
$F^2=\lambda$ and thus $a^1_6=a^1_8=0$ and $F^2_{tt}=-a^1_7F^2$.
The template-form equation~\eqref{eq:2DSWEsSysOfTemplateFormEqs} reduces to
$b_x+a^1_5b+a^1_7x+a^1_9=0$.

If $a^1_5=0$, then we can set $a^1_7\in\{0,-1,1\}$ by using scaling equivalence transformations,
which leads to Cases~\ref{case7}, \ref{case8} and~\ref{case9}, respectively.
Note that $a^1_9=0\bmod G^\sim$ if $a^1_7\ne0$
and $a^1_9\in\{0,-1\}\bmod G^\sim$ if $a^1_7=0$.

If $a^1_5\ne0$, then $\lambda$ is a constant, and thus $a^1_7=0$.
We can again set $a^1_9=0$ by shifts with respect to $b$
as well as $a^1_5=1$ up to scaling equivalence transformations and alternating signs of~$(x,u)$.
This leads to Case \ref{case5}.

\subsection{Two independent template-form equations}

\looseness=-1
For $k=2$, the right-hand side of the equation~\eqref{eq:2DSWEsClassifyingEq} is
a linear combination of right-hand sides of the first and the second equations of the system~\eqref{eq:2DSWEsSysOfTemplateFormEqs}
with coefficients~$\lambda^1$ and~$\lambda^2$ that depend on~$t$,
\begin{gather*}
F^1_t(xb_x+yb_y)+c_2(yb_x-xb_y)+F^2b_x+F^3b_y+2(F^1_t-c_1)b\\
\qquad{}-F^1_{ttt}\frac{x^2+y^2}2-F^2_{tt}x-F^3_{tt}y-F^4\\
=\sum_{i=1}^2\lambda^i\bigg(a^i_1(xb_x+yb_y)+a^i_2(yb_x-xb_y)+a^i_3b_x+a^i_4b_y+a^i_5b\\
\qquad{}+a^i_6\frac{x^2+y^2}2+a^i_7x+a^i_8y+a^i_9\bigg).
\end{gather*}

\begin{remark}\label{rem:2DSWEsOnNonproportionalityOfMultipliersOfTemplate-formEqs}
The coefficients~$\lambda^1$ and~$\lambda^2$ are not proportional
with the same constant multiplier for all vector fields from~$\mathfrak g_b$
since otherwise there is no additional Lie-symmetry extension in comparison with
the more general case of Lie-symmetry extension with $k=1$,
where the corresponding linear combination of equations of the system~\eqref{eq:2DSWEsSysOfTemplateFormEqs}
plays the role of a single template-form equation.
For the same reason, both these coefficients do not vanish identically for some vector fields from~$\mathfrak g_b$.
\end{remark}

Splitting the resulting condition with respect to $b$, its derivatives $b_x$ and $b_y$ 
and the independent variables~$x$ and~$y$, we derive the system
\begin{equation}\label{gen_restriction}
\begin{split}
&F^1_t=a^1_1\lambda^1+a^2_1\lambda^2,\quad c_2=a^1_2\lambda^1+a^2_2\lambda^2,\quad F^2=a^1_3\lambda^1+a^2_3\lambda^2, \quad F^3=a^1_4\lambda^1+a^2_4\lambda^2,\\
&2(F^1_t-c_1)=a^1_5\lambda^1+a^2_5\lambda^2, \quad F^1_{ttt}=-a^1_6\lambda^1-a^2_6\lambda^2,\\
&F^2_{tt}=-a^1_7\lambda^1-a^2_7\lambda^2,\quad F^3_{tt}=-a^1_8\lambda^1-a^2_8\lambda^2,\quad F^4=-a^1_9\lambda^1-a^2_9\lambda^2.
\end{split}
\end{equation}

The further consideration for $k=2$ is partitioned into different cases
depending on the rank of the submatrix $A_2$ that is constituted by the first two columns of~$A$.
Since $\rank A_2\leqslant2$, we have the cases $\rank A_2=2$, $\rank A_2=1$ and $\rank A_2=0$.

\medskip\par\noindent $\boldsymbol{\rank A_2=2.}$
Linearly re-combining equations of the system~\eqref{eq:2DSWEsSysOfTemplateFormEqs},
we can set the matrix $A_2$ to be the identity matrix, i.e., $a^1_1=a^2_2=1$ and $a^1_2=a^2_1=0$.
To further simplify the form of the system~\eqref{eq:2DSWEsSysOfTemplateFormEqs},
we set $a^1_3=a^1_4=0$ by equivalence transformations of shifts with respect to $x$ and~$y$.
In view of the condition~\eqref{eq:CompatibilityCondInTermsOfAssociatedVFs},
the vector fields~$\mathbf v_1$ and~$\mathbf v_2$,
which are associated with the first and the second equations of the reduced system~\eqref{eq:2DSWEsSysOfTemplateFormEqs}, respectively, commute.
This yields the system of algebraic equations with respect to the coefficients~$a^i_j$,
\begin{equation}\label{alg_condition}
\begin{split}
&a^2_3=a^2_4=0,\quad
 a^1_7-a^2_8+a^1_8a^2_5-a^1_5a^2_8=0,\quad
 a^1_8+a^2_7-a^1_7a^2_5+a^1_5a^2_7=0,\\
&2a^2_6-a^1_6a^2_5+a^1_5a^2_6=0,\quad
 a^1_9a^2_5-a^1_5a^2_9=0.
\end{split}
\end{equation}
The reduced form of the system~\eqref{gen_restriction} is
\begin{equation}\label{restrict2.2}
\begin{split}
&F^1_t=\lambda^1,\quad c_2=\lambda^2,\quad F^2=0, \quad F^3=0,\quad (a^1_5-2)\lambda^1=-2c_1-a^2_5c_2, \\
&\lambda^1_{tt}+a^1_6\lambda^1+a^2_6\lambda^2=0,\quad a^1_7\lambda^1+a^2_7\lambda^2=0,\quad a^1_8\lambda^1+a^2_8\lambda^2=0,\quad F^4=-a^1_9\lambda^1-a^2_9\lambda^2.
\end{split}
\end{equation}
In view of Remark~\ref{rem:2DSWEsOnNonproportionalityOfMultipliersOfTemplate-formEqs},
the seventh and the eighth equations of the system~\eqref{restrict2.2}
imply $a^1_7=a^2_7=0$ and $a^1_8=a^2_8=0$, respectively.
The reduced form of the system~\eqref{eq:2DSWEsSysOfTemplateFormEqs} in the polar coordinates~$(r,\varphi)$ is
$rb_r+a^1_5b+\frac12a^1_6r^2+a^1_9=0$, $b_\varphi=a^2_5b+\frac12a^2_6r^2+a^2_9$.
In view of the last equation of~\eqref{alg_condition},
up to shifts of~$b$
we can set $a^1_9=a^2_9=0$ if $(a^1_5,a^2_5)\ne(0,0)$.

The further consideration depends on whether or not the parameter~$a^1_5$ is equal to 2,
and if it is not, whether or not the parameter~$a^2_5$ is zero.

If $a^1_5=2$, then we can assume the parameter~$a^1_6$ to belong to $\{0,-4,4\}$,
and $a^2_6=a^2_5a^1_6/4$.
Integrating the corresponding system~\eqref{eq:2DSWEsSysOfTemplateFormEqs},
we find the general form of the arbitrary element~$b$,
\begin{gather*}
b=b_0r^{-2}\exp(a^2_5\varphi)-\frac{a^1_6}8r^2,
\end{gather*}
where the integration constant~$b_0$ is nonzero
since otherwise this value of~$b$ is associated with the value $k=5$.
This is why we can scale~$b_0$ by an equivalence transformation to $\varepsilon=\pm1$,
which leads, depending on the value of~$a^1_6$,
to Cases~\ref{case12a},~\ref{case12b} and \ref{case12c} of Theorem~\ref{thm:GroupClassificationOf2DSWEs1}.

If $a^1_5\ne2$, then $\lambda^1$ is a constant.
Then the sixth equation of the system~\eqref{restrict2.2} takes the form $a^1_6\lambda^1+a^2_6\lambda^2=0$,
implying, according to Remark~\ref{rem:2DSWEsOnNonproportionalityOfMultipliersOfTemplate-formEqs},
that $a^1_6=a^2_6=0$.
Depending on whether $(a^1_5,a^2_5)\ne(0,0)$ or $a^1_5=a^2_5=0$,
we obtain Cases~\ref{case10} and~\ref{case11} of Theorem~\ref{thm:GroupClassificationOf2DSWEs1}, respectively.

\par\medskip\par\noindent $\boldsymbol{\rank A_2=1.}$
Linearly recombining equations of the system~\eqref{eq:2DSWEsSysOfTemplateFormEqs},
we reduce the matrix $A_2$ to the form
\begin{gather*}
A_2=\begin{pmatrix}
a^1_1&a^1_2\\
0&0
\end{pmatrix},
\end{gather*}
i.e., $a^2_1=a^2_2=0$ and $(a^1_1,a^1_2)\ne(0,0)$.
We also have $(a^2_3,a^2_4)\ne (0,0)$ since $\rank A_4=2$.
Hence we can set $a^2_3=1$, $a^2_4=0$ by a rotation equivalence transformation
and re-scaling the second equation and make $a^1_3=a^1_4=0$ by shifts of $x$ and $y$.
The condition~\eqref{eq:CompatibilityCondInTermsOfAssociatedVFs} implies that $a^1_2=0$ and hence $a^1_1\ne0$,
so we can set $a^1_1=1$ by rescaling of the first equation.
Then the condition~\eqref{eq:CompatibilityCondInTermsOfAssociatedVFs} is equivalent to the commutation relation $[\mathbf v_1,\mathbf v_2]=-\mathbf v_2$,
yielding the following system of algebraic equations on the remaining coefficients~$a^i_j$:
\begin{equation}\label{rk1_alg_restriction}
a^2_5=0,\quad
a^2_6(a^1_5+3)=0,\quad
a^2_8(a^1_5+2)=0,\quad
a^2_7(a^1_5+2)=a^1_6,\quad
a^2_9(a^1_5+1)=a^1_7.
\end{equation}
The system~\eqref{gen_restriction} is simplified to
\begin{equation}\label{rk1_split}
\begin{split}
&F^1_t=\lambda^1,\quad c_2=0,\quad F^2=\lambda^2,\quad F^3=0,\quad 2c_1=(2-a^1_5)\lambda^1,\\
&F^1_{ttt}=-a^1_6\lambda^1-a^2_6\lambda^2,\quad F^2_{tt}=-a^1_7\lambda^1-a^2_7\lambda^2,\quad
a^1_8\lambda^1+a^2_8\lambda^2=0,\quad F^4=-a^1_9\lambda^1-a^2_9\lambda^2.
\end{split}
\end{equation}
In view of Remark~\ref{rem:2DSWEsOnNonproportionalityOfMultipliersOfTemplate-formEqs},
the eighth equation of the system~\eqref{rk1_split} implies $a^1_8=a^2_8=0$.

Supposing that $a^2_6\ne0$, we successively derive from the second equation of~\eqref{rk1_alg_restriction}
and the system~\eqref{rk1_split} that $a^1_5=-3$, $\lambda^1$ is a constant,
$a^1_6\lambda^1+a^2_6\lambda^2=0$ and, according to Remark~\ref{rem:2DSWEsOnNonproportionalityOfMultipliersOfTemplate-formEqs}, $a^1_6=a^2_6=0$,
which is a contradiction. Hence $a^2_6=0$.

It is obvious from the fifth equation of~\eqref{rk1_split} that the value $a^1_5=2$ is special.
For $a^1_5\ne2$, we obtain $\lambda^1=\const$, $F^1_{tt}=0$, $a^1_6=0$,
and thus the fourth equation of the system~\eqref{rk1_alg_restriction} takes the form $(a^1_5+2)a^2_7=0$,
implying $a^2_7=0$ if $a^1_5\ne-2$.
Therefore, the value $a^1_5=-2$ is special as well.
In the course of integrating the system~\eqref{eq:2DSWEsSysOfTemplateFormEqs}, the value $a^1_5=0$ is additionally singled out.
Moreover, if $a^1_5\ne0$, we can make $a^1_9=0$ using a shift of~$b$.
As a result, we need to separately consider each of the above values $2$, $0$, $-2$ of~$a^1_5$
and the case $a^1_5\notin\{-2,0,2\}$.

\medskip\par\noindent\textbf{1.}
$a^1_5=2$. Shifting $b$, we set $a^1_9=0$.
The system~\eqref{rk1_alg_restriction} reduces to the equations $a^2_7=a^1_6/4$ and $a^2_9=a^1_7/3$.

Let $a^1_6\ne0$.
Then shifting $x$ and~$b$ and recombining equations of the system~\eqref{eq:2DSWEsSysOfTemplateFormEqs},
we also set $a^1_7=0$, and consequently $a^2_9=0$.
The general solution to the system~\eqref{eq:2DSWEsSysOfTemplateFormEqs} is
$b=b_0y^{-2}-\frac18a^1_6(x^2+y^2)$,
where the integration constant~$b_0$ is nonzero 
since otherwise this value of the arbitrary element~$b$ is associated with $k=5$.
Using scaling equivalence transformations, we can set $b_0,a^1_6/4\in\{-1,1\}$.
Depending on the sign of $a^1_6$,
we obtain Cases~\ref{case16b} and~\ref{case16c} of Theorem~\ref{thm:GroupClassificationOf2DSWEs1}.

If $a^1_6=0$, then $a^2_7$ is also zero. Integrating~\eqref{eq:2DSWEsSysOfTemplateFormEqs}, we get
$b=b_0y^{-2}-a^2_9x$,
where again the integration constant~$b_0$ is nonzero since otherwise this value of the arbitrary element~$b$
is associated with $k=5$.
Using scaling equivalence transformations and alternating the signs of $(x,u)$,
we can set $b_0\in\{-1,1\}$ and $a^2_9\in\{0,-1\}$,
which gives Case~\ref{case16a}.

\medskip\par\noindent\textbf{2.}
$a^1_5=0$. Then $a^2_7=0$ and $a^1_7=a^2_9$.
The system~\eqref{eq:2DSWEsSysOfTemplateFormEqs} integrates to
$b=-a^1_9\ln |y|-a^2_9x+b_0$,
where the parameter~$a^1_9$ is nonzero since otherwise $k=5$.
The integration constant~$b_0$ can be set to zero by shifts of~$b$,
as well as $a^1_6\in\{-1,1\}$ and $a^2_9\in\{-1,0\}$
up to scaling equivalence transformations and alternating the signs of $(x,u)$.
This leads to Case~\ref{case14}.

\medskip\par\noindent\textbf{3.}
$a^1_5=-2$. Then $a^1_7=-a^2_9$.
We shift~$b$ for setting $a^1_9=0$.
The general solution of the system~\eqref{eq:2DSWEsSysOfTemplateFormEqs} is
$b=b_0y^2-\frac12a^2_7x^2-a^2_9x$
but this value of the arbitrary element $b$ is associated with $k>2$.

\medskip\par\noindent\textbf{4.}
$a^1_5\notin\{-2,0,2\}$.
Solving the system~\eqref{eq:2DSWEsSysOfTemplateFormEqs}, we obtain
$b=b_0|y|^{-a^1_5}-a^2_9x$,
where the integration constant~$b_0$ has be nonzero for $k=2$.
Setting $b_0\in\{-1,1\}$ and $a^2_9\in\{-1,0\}$
by scaling equivalence transformations and alternating the signs of $(x,u)$
results in Case~\ref{case13}.

\medskip\par\noindent $\boldsymbol{\rank A_2=0.}$
This means that $a^1_1=a^1_2=a^2_1=a^2_2=0$.
Since $\rank A_4=2$, we can linearly recombine equations of the system~\eqref{eq:2DSWEsSysOfTemplateFormEqs}
to set $a^1_3=a^2_4=1$ and $a^1_4=a^2_3=0$.
The compatibility condition~\eqref{eq:CompatibilityCondInTermsOfAssociatedVFs} means
that the vector fields~$\mathbf v_1$ and~$\mathbf v_2$ associated to equations of the reduced system~\eqref{eq:2DSWEsSysOfTemplateFormEqs}
commutes, $[\mathbf v_1,\mathbf v_2]=0$, which results in the system
\begin{equation}\label{eq:2DSWEsK1RankA20SysForA}
\begin{split}
&a^2_6-a^1_7a^2_5+a^1_5a^2_7=0,\quad
 a^1_6a^2_5-a^1_5a^2_6=0,\\
&a^1_6+a^1_8a^2_5-a^1_5a^2_8=0,\quad
a^1_8-a^2_7+a^1_9a^2_5-a^1_5a^2_9=0.
\end{split}
\end{equation}

Suppose that $a^1_5=a^2_5=0$.
The system~\eqref{eq:2DSWEsK1RankA20SysForA} then reduces to $a^1_6=a^2_6=0$ and $a^1_8=a^2_7$.
In view of the last equation, we can set $a^1_8=a^2_7=0$ by rotation equivalence transformations.
The general solution of the system~\eqref{eq:2DSWEsSysOfTemplateFormEqs} is
\begin{gather*}
b=-\frac{a^1_7}2 x^2-\frac{a^2_8}2 y^2-a^1_9x-a^2_9y+b_0,
\end{gather*}
where $b_0$ is an integration constant.
This form of the arbitrary element $b$ is related to the value $k=3$ if $a^1_7\ne a^2_8$
and to the value $k=5$ if $a^1_7=a^2_8$, which contradicts the supposition $k=2$.

This is why $(a^1_5,a^2_5)\ne(0,0)$, and we set $a^1_5=0$ and $a^2_5=-1$
by equivalence transformations of rotations, scalings and alternating signs.
Then the first equation of the system~\eqref{eq:2DSWEsK1RankA20SysForA} implies $a^1_6=0$,
and thus the first six equations of the system~\eqref{gen_restriction} take the form
$F^1_t=0$, $c_2=0$, $F^2=\lambda^1$, $F^3=\lambda^2=2c_1$ and $a^2_6\lambda^2=0$.
In view of the last equation, we get $a^2_6=0$.
Therefore, the system~\eqref{eq:2DSWEsK1RankA20SysForA} is equivalent to $a^1_7=a^1_8=0$ and $a^2_7=-a^1_9$.
The eighth equation of~\eqref{gen_restriction} gives $a^2_8\lambda^2=0$, i.e., $a^2_8=0$.
We can set $a^2_9=0$ up to equivalence transformations of shifts of $b$.
The system~\eqref{eq:2DSWEsSysOfTemplateFormEqs} integrates to
$b=b_0e^y+a^2_7x$,
where the integration constant~$b_0$ is nonzero
since otherwise $b$ is a linear function, for which $k=5$.
Equivalence transformations of scalings and alternating the signs of~$(x,u)$
allow us to set $b_0=\pm1$ and $a^2_7\in\{0,1\}$.
This results in Case~\ref{case15} of Theorem~\ref{thm:GroupClassificationOf2DSWEs1}.

\subsection{More independent template-form equations}

We show below that $k>2$ if and only if $b$ is at most quadratic polynomial in~$(x,y)$.

\medskip\par\noindent$\boldsymbol{k=3.}$
Since $\rank A_4=\rank A=k=3$, we have that $\rank A_2>0$.

Suppose that the submatrix $A_2$ is of rank two.
Recombining equations of the system~\eqref{eq:2DSWEsSysOfTemplateFormEqs}, we can reduce this submatrix to the form
\[
A_2=\begin{pmatrix}
1&0\\
0&1\\
0&0
\end{pmatrix}.
\]
Then the projections $\hat{\mathbf v}_1$, $\hat{\mathbf v}_2$ and $\hat{\mathbf v}_3$
of the vector fields $\mathbf v_1$, $\mathbf v_2$ and $\mathbf v_3$
to the space with the coordinates $(x,y)$ are
\begin{gather}\label{k3_vf_part}
\hat{\mathbf v}_1=(x+a^1_3)\p_x+(y+a^1_4)\p_y,\quad
\hat{\mathbf v}_2=(y+a^2_3)\p_x-(x-a^2_4)\p_y,\quad
\hat{\mathbf v}_3=a^3_3\p_x+a^3_4\p_y.
\end{gather}
In view of the condition~\eqref{eq:CompatibilityCondInTermsOfProjectionsOfAssociatedVFs},
the commutator $[\hat{\mathbf v}_2,\hat{\mathbf v}_3]=-a^3_4\p_x+a^3_3\p_y$
should belong to the span $\langle\hat{\mathbf v}_1,\hat{\mathbf v}_2,\hat{\mathbf v}_3\rangle$
but this is not the case, which is a contradiction.

Therefore, $\rank A_2=1$, and thus the matrix $A_4$ can be reduced to the form
\[
A_4=\begin{pmatrix}
a^1_1&a^1_2&0&0\\
0&0&1&0\\
0&0&0&1
\end{pmatrix},
\quad\mbox{where}\quad (a^1_1,a^1_2)\ne(0,0).
\]
Then the compatibility condition~\eqref{eq:CompatibilityCondInTermsOfAssociatedVFs} is equivalent to
the commutation relations
\[
[\mathbf v_1,\mathbf v_2]=-a^1_1\mathbf v_2+a^1_2\mathbf v_3,\quad
[\mathbf v_1,\mathbf v_3]=-a^1_1\mathbf v_3-a^1_2\mathbf v_2,\quad
[\mathbf v_2,\mathbf v_3]=0.
\]
From the first two commutation relations, we obtain $a^1_1a^2_5-a^1_2a^3_5=0$, $a^1_2a^2_5+a^1_1a^3_5=0$,
and thus $a^2_5=a^3_5=0$ since $(a^1_1,a^1_2)\ne(0,0)$.
Then the last commutation relation yields $a^2_6=a^3_6=0$ and $a^3_7=a^2_8$.
Up to rotation equivalence transformations, we can set $a^2_8=a^3_7=0$.
Under this gauge, from the former commutation relations we get the system
\begin{gather}\label{eq:2DSWEsK3SysForA}
\begin{split}
&(2a^1_1+a^1_5)(a^2_7-a^3_8)=0,\quad
 a^1_2(a^2_7-a^3_8)=0,\\
&a^1_6=(2a^1_1+a^1_5)a^2_7,\quad
 a^1_7=(a^1_1+a^1_5)a^2_9-a^1_2a^3_9,\quad
 a^1_8=(a^1_1+a^1_5)a^3_9+a^1_2a^2_9.
\end{split}
\end{gather}
Since the arbitrary element~$b$ satisfies the equations $b_x+a^2_7x+a^2_9=0$ and $b_y+a^3_8y+a^3_9=0$,
it is a quadratic function of $(x,y)$.
More specifically,
\begin{gather}\label{eq:2DSWEsK3FormOfB}
b=-\frac12a^2_7x^2-\frac12a^3_8y^2-a^2_9x-a^3_9y
\end{gather}
up to equivalence transformations of shifts with respect to~$b$.
Hence $a^2_7\ne a^3_8$ since otherwise $k=5$ for this value of~$b$,
which contradicts the supposition $k=3$.
Then the system~\eqref{eq:2DSWEsK3SysForA} reduces~to
\begin{gather*}
a^1_2=0,\quad a^1_5=-2a^1_1,\quad a^1_6=0,\quad a^1_7=-a^1_1a^2_9,\quad a^1_8=-a^1_1a^3_9
\end{gather*}
and guarantees that the above value of~$b$ satisfies the entire corresponding system~\eqref{eq:2DSWEsSysOfTemplateFormEqs}.

Modulo $G^\sim$-equivalence, we can assume that $a^2_7=\pm 1$, $a^2_9=0$; $|a^3_8|<|a^2_7|$ if $a^2_7a^3_8>0$;
$a^3_9=0$ if $a^3_8\ne0$; $a^3_9\in\{-1,0\}$ if $a^3_8=0$.
$G^\sim$-inequivalent values of~$b$ of the form~\eqref{eq:2DSWEsK3FormOfB}
with the associated maximal Lie invariance algebras are listed in
Cases~\ref{case17}--\ref{case21} of Theorem~\ref{thm:GroupClassificationOf2DSWEs1}.

\medskip\par\noindent$\boldsymbol{k=4.}$
Since $\rank A_4=\rank A=4$, by linearly re-combining the equations~\eqref{eq:2DSWEsSysOfTemplateFormEqs},
we can set $A_4$ to be the $4\times4$ identity matrix.
In view of the form of the vector fields~$\mathbf v_1$, \dots, $\mathbf v_4$
associated to the equations of the reduced system~\eqref{eq:2DSWEsSysOfTemplateFormEqs},
\noprint{
\begin{gather*}
\mathbf v_1=x\p_x+y\p_y-\left(a^1_5b+\tfrac12a^1_6(x^2+y^2)+a^1_7x+a^1_8y+a^1_9\right)\p_b,\\
\mathbf v_2=y\p_x-x\p_y-\left(a^2_5b+\tfrac12a^2_6(x^2+y^2)+a^2_7x+a^2_8y+a^2_9\right)\p_b,\\
\mathbf v_3=\p_x-\left(a^3_5b+\tfrac12a^3_6(x^2+y^2)+a^3_7x+a^3_8y+a^3_9\right)\p_b,\\
\mathbf v_4=\p_y-\left(a^4_5b+\tfrac12a^4_6(x^2+y^2)+a^4_7x+a^4_8y+a^4_9\right)\p_b,
\end{gather*}
}
the compatibility condition~\eqref{eq:CompatibilityCondInTermsOfAssociatedVFs} implies
the commutation relations
\[
[\mathbf v_2,\mathbf v_3]=\mathbf v_4,\quad
[\mathbf v_2,\mathbf v_4]=-\mathbf v_3,\quad
[\mathbf v_3,\mathbf v_4]=0.
\]
From the first two commutation relations, we get that $a^3_5=a^4_5=0$.
The last commutation relation together with the previous restrictions on the coefficients $a^i_j$ yields $a^3_6=a^4_6=0$ and $a^3_8=a^4_7$.
Returning to the first two commutation relations, we obtain the equations
\[
a^2_6=-2a^3_8+a^2_5a^3_7=2a^3_8+a^2_5a^4_8, \quad a^3_7-a^4_8+a^2_5a^3_8=0
\]
implying $a^3_8=0$ and $a^3_7=a^4_8$.
Since the arbitrary element~$b$ satisfies the equations $b_x+a^3_7x+a^3_9=0$ and $b_y+a^3_7y+a^4_9=0$,
it is a quadratic function of $(x,y)$ with the same coefficients of $x^2$ and of $y^2$ 
and with zero coefficient of~$xy$.
This means that in fact $k=5$, which contradicts the supposition $k=4$.

\medskip\par\noindent$\boldsymbol{k=5.}$
The $5\times 5$ matrix~$A_5$ of the coefficients of the system~\eqref{eq:2DSWEsSysOfTemplateFormEqs} is of rank $5$
and, up to recombining equations of this system, can be assumed to be the $5\times5$ identity matrix.
Then the last equation of the system~\eqref{eq:2DSWEsSysOfTemplateFormEqs} implies
that $b$ is the specific quadratic polynomial of $(x,y)$,
\[
b=-\tfrac12a^5_6(x^2+y^2)-a^5_7x-a^5_8y-a^5_9.
\]
There are four $G^\sim$-inequivalent values of the arbitrary element~$b$ among such quadratic polynomials,
$b=0$, $b=x$, $b=\tfrac12(x^2+y^2)$ and $b=-\tfrac12(x^2+y^2)$,
which correspond to Cases~\ref{case22a}, \ref{case22b}, \ref{case22c} and~\ref{case22d}
of Theorem~\ref{thm:GroupClassificationOf2DSWEs1}, respectively.

\section[Additional equivalence transformations and modified classification result]
{Additional equivalence transformations\\ and modified classification result}\label{sec:AdditionalEquivalencesSWE}

It is obvious that the class~\eqref{eq:2DSWEs} is not normalized.
In other words, it possesses admissible transformations
(e.g., those related to Lie symmetries of systems from this class)
that are not generated by elements of~$G^\sim$.
See \cite{bihl11Dy,kuru18a,opan17a,popo06Ay,popo10Ay,vane2020a} for definitions.
Moreover, we will show below that the class~\eqref{eq:2DSWEs} is not even semi-normalized
since some of its admissible transformations cannot be presented as compositions
of those generated by elements of~$G^\sim$ and those generated  by point symmetries of systems from this class.
Such admissible transformations may lead to additional point equivalences
among classification cases listed in Theorem~\ref{thm:GroupClassificationOf2DSWEs1}.

Since we do not have the complete description of the equivalence groupoid of the class~\eqref{eq:2DSWEs},
in the course of looking for the above additional equivalences,
we need to use algebraic tools that do not rely on this description.
If two systems of differential equations are similar with respect to a point transformation,
then the corresponding maximal Lie invariance algebras are isomorphic in the sense of abstract Lie algebras.
Moreover, these maximal Lie invariance algebras are similar
as realizations of Lie algebras by vector fields with respect to the same point transformation.
This gives necessary conditions of similarity for systems of differential equations with respect to point transformations.

Lie algebras of different dimensions are nonisomorphic.
Hence we categorize the Lie algebras presented in Theorem~\ref{thm:GroupClassificationOf2DSWEs1}
according to their dimensions to distinguish the cases that are definitely not equivalent to each other
with respect to point transformations,
\begin{itemize}\itemsep-.5ex
\item \ $\dim\mathfrak g_b=2$: Cases~\ref{case1}, \ref{case2}, \ref{case3}, \ref{case4} and \ref{case5};
\item \ $\dim\mathfrak g_b=3$: Cases~\ref{case6a}, \ref{case6b}, \ref{case6c}, \ref{case7}, \ref{case8}, \ref{case9}, \ref{case10} and \ref{case11};
\item \ $\dim\mathfrak g_b=4$: Cases~\ref{case12a}, \ref{case12b}, \ref{case12c}, \ref{case13}, \ref{case14} and~\ref{case15};
\item \ $\dim\mathfrak g_b=5$: Cases~\ref{case16a}, \ref{case16b} and~\ref{case16c};
\item \ $\dim\mathfrak g_b=6$: Cases~\ref{case17}, \ref{case18}, \ref{case19}, \ref{case20} and~\ref{case21};
\item \ $\dim\mathfrak g_b=9$: Cases~\ref{case22a}, \ref{case22b}, \ref{case22c} and~\ref{case22d}.
\end{itemize}
However, the same dimension of algebras does not ensure their isomorphism.

Finding a pair of classification cases with isomorphic maximal Lie invariance algebras
and fixing bases of these algebras that are concordant under the found algebra isomorphism,
we aim to obtain a point transformation that respectively maps the basis elements of the first algebra to the basis elements of the second one.
The existence of such a point transformation hints that the two classification cases may be equivalent with respect to this very transformation.

In this way, we find three families of $G^\sim$-inequivalent admissible transformations for the class~\eqref{eq:2DSWEs}
that are not induced by equivalence transformations of this class and whose target arbitrary elements differ from their source arbitrary elements.
In each of these families, the source arbitrary elements are parameterized by an arbitrary function of a single argument.
We present these families jointly with
the corresponding induced additional equivalences between classification cases of Theorem~\ref{thm:GroupClassificationOf2DSWEs1}:

\begin{enumerate}
\item
$b=r^{-2}f(\varphi)-\frac12r^2$, \ $\tilde b=\tilde r^{-2}f(\tilde\varphi)$,\\[1ex]
$\tilde t=\tan t$, \ $\tilde x=x\sec t$, \ $\tilde y=y\sec t$, \ $\tilde u=u\cos t+x\sin t$, \ $\tilde v=v\cos t+y\sin t$, \ $\tilde h=h\cos^2 t$,\\[1ex]
\ref{case6c} $\to$ \ref{case6a}, \
\ref{case12c} $\to$ \ref{case12a}, \
\ref{case16b} $\to$ \ref{case16a}$_{\delta=0}$, \
\ref{case22d}  $\to$ \ref{case22a}.
\item
$b=r^{-2}f(\varphi)+\frac12r^2$, \ $\tilde b=\tilde r^{-2}f(\tilde\varphi)$,\\[1ex]
$\tilde t=\tfrac12e^{2t}$, \ $\tilde x=e^tx$, \ $\tilde y=e^ty$, \ $\tilde u=e^{-t}\left( u+x\right)$, \ $\tilde v=e^{-t}\left( v+y\right)$, \ $\tilde h=e^{-2t}h$,\\[1ex]
\ref{case6b} $\to$ \ref{case6a}, \
\ref{case12b} $\to$ \ref{case12a}, \
\ref{case16c} $\to$ \ref{case16a}$_{\delta=0}$, \
\ref{case22c}  $\to$ \ref{case22a}.
\item
$b=f(y)+x$, \ $\tilde b=f(\tilde y)$,\\[1ex]
$\tilde t=t$, \ $\tilde x=x+\tfrac12t^2$, \ $\tilde y=y$, \ $\tilde u=u+t$, \ $\tilde v=v$, \ $\tilde h=h$,\\[1ex]
\ref{case22b}  $\to$ \ref{case22a}, \quad
\ref{case7}, \ref{case13}, \ref{case14}, \ref{case15}, \ref{case16a}, \ref{case18}, \ref{case20}:%
\footnote{%
For Cases~\ref{case18} and~\ref{case20}, we should compose the corresponding point transformation with
the permutation $(x,u)\leftrightarrow(y,v)$, 
which is related to an equivalence transformations of the class~\eqref{eq:2DSWEs}.
}
\ $\delta=1$ $\to$ $\delta=0$.
\end{enumerate}

The first and the second families of admissible transformations can be generalized to the rotating reference frame.
The generalization of the transformation with $f=0$ from the first family
to the rotating reference frame was found for the first time
for the shallow water equations in cylindrical coordinates in~\cite[Theorem 1]{ches11a}.

Each of the above admissible transformations is $G^\sim$-equivalent to no admissible transformation
generated by point symmetries of systems from this class.
Therefore, the class~\eqref{eq:2DSWEs} is not semi-normalized.

As a by-product, we also prove the following assertions.

\begin{proposition}
Any system from the class~\eqref{eq:2DSWEs} that is invariant with respect to a nine-dimensional Lie algebra of vector fields
is equivalent, up to point transformations, to the system from the same class with~$b=0$,
which is the system of shallow water equations with flat bottom topography.
\end{proposition}

\begin{proposition}\looseness=-1
Any system from the class~\eqref{eq:2DSWEs} with five-dimensional maximal Lie invariance algebra
is reduced by a point transformation to the system from the same class with~$b=\pm y^{-2}$.
\end{proposition}

For the other possible dimensions of maximal Lie invariance algebras of systems from the class~\eqref{eq:2DSWEs},
we prove the inequivalence of the remaining classification cases whenever it is possible to do so using the algebraic technique
based on Mubarakzianov's classification of Lie algebras up to dimension four~\cite{muba63a}
and Turkowski's classification of six-dimensional solvable Lie algebras with four-dimensional nilradicals~\cite{turk90a}.
For convenience, we take these classifications in the form given in~\cite{boyk06a}.
There, the notation of algebras from Mubarakzianov's classification was modified, in particular,
by indicating parameters for families of algebras,
and the basis elements of the algebras from Turkowski's classification
were renumbered in order to have bases in $K$-canonical forms.
For each abstract Lie algebra appearing in the consideration,
we present all the nonzero commutation relations among basis elements up to antisymmetry.

Unfortunately, the algebraic criterion of inequivalence with respect to point transformations
is not sufficiently powerful for systems with two-dimensional maximal Lie invariance algebra
since there are only two nonisomorphic two-dimensional Lie algebras~$2A_1$ and~$A_{2.1}$, the abelian and the non-abelian ones.
Applying this criterion, we can only partition the corresponding classification cases into two sets,
Cases~\ref{case1}$_{\nu=2}$ and \ref{case3} with abelian two-dimensional Lie invariance algebras and
Cases~\ref{case1}$_{\nu\ne2}$, \ref{case2}, \ref{case4} and \ref{case5} with non-abelian two-dimensional Lie invariance algebras.
There are definitely no point transformations between cases that belong to different sets.

For the classification cases with maximal Lie invariance algebras of greater dimensions,
the algebraic criterion is more advantageous.
Thus, the maximal Lie invariance algebras
in Cases~\ref{case6a}, \ref{case7}$_{\delta=0}$, \ref{case8}, \ref{case9}, \ref{case10} and~\ref{case11}
of Theorem~\ref{thm:GroupClassificationOf2DSWEs1}
are isomorphic to the three-dimensional Lie algebras
${\rm sl}(2,\mathbb R)$, $A_{3.1}$, $A_{3.4}^{-1}$, $A_{3.5}^0$, $A_{2.1}\oplus A_1$ and $A_{2.1}\oplus A_1$,
respectively.
These Lie algebras are defined by the following commutation relations:\par
\bigskip
\begin{tabular}{ll}
${\rm sl}(2,\mathbb R)$: & $[e_1,e_2]=e_1,\ [e_2,e_3]=e_3,\ [e_1,e_3]=-2e_2$;\\[1ex]
$A_{3.1}$:               & $[e_2,e_3]=e_1$;\\[1ex]
$A_{3.4}^{-1}$:          & $[e_1,e_3]=e_1,\ [e_2,e_3]=-e_2$;\\[1ex]
$A_{3.5}^0$:             & $[e_1,e_3]=-e_2,\ [e_2,e_3]=e_1$;\\[1ex]
$A_{2.1}\oplus A_1$:     & $[e_1,e_2]=e_1$,
\end{tabular}
\bigskip\par
\noindent
and they are well known to be non-isomorphic to each other.
This implies the pairwise inequivalence of the above cases of Lie-symmetry extensions with respect to point transformations,
except the pair of Cases~\ref{case10} and~\ref{case11}.

In a similar way, we prove that the maximal Lie invariance algebras
associated with Cases~\ref{case12a}, \ref{case13}$_{\delta=0}$, \ref{case14}$_{\delta=0}$ and~\ref{case15}$_{\delta=0}$
are also not isomorphic to each other although all of them are four-dimensional.
The corresponding abstract Lie algebras are ${\rm sl}(2,\mathbb R)\oplus A_1$ (Case~\ref{case12a}) as well as
$A_{4,8}^a$ with $a=\nu/(2-\nu)$ if $\nu<1$ and with $a=(2-\nu)/\nu$ if $\nu\geqslant1$ (Case~\ref{case13}),
with $a=0$ (Case~\ref{case14}) and with $a=-1$ (Case~\ref{case15}),
where\par
\bigskip
\begin{tabular}{ll}
$A_{4.8}^a$, $|a|\leqslant 1$:\quad & $[e_2,e_3]=e_1$,  $[e_1,e_4]=(1+a)e_1$, $[e_2,e_4]=e_2$, $[e_3,e_4]=ae_3$.
\end{tabular}
\bigskip\par
\noindent
As a result, Cases~\ref{case12a}, \ref{case13}, \ref{case14} and~\ref{case15} of Theorem~\ref{thm:GroupClassificationOf2DSWEs1} are equivalent neither to each other nor to other cases of this theorem with respect to point transformations.
Moreover, the parameter~$\nu$ in Case~\ref{case13} cannot be gauged by point transformations.

In Cases~\ref{case17}, \ref{case18}$_{\delta=0}$, \ref{case19}, \ref{case20}$_{\delta=0}$ and~\ref{case21},
the corresponding maximal Lie invariance algebras are six-dimensional solvable Lie algebras with four-dimensional abelian nilradicals
that are respectively isomorphic to the algebras
$N_{6.1}^{abcd}$ with $(a,b,c,d)=\frac12(1-\beta,1+\beta,1+\beta,1-\beta)$,
$N_{6.2}^{-1,1,1}$,
\smash{$N_{6.13}^{-1/\beta,1/\beta,1,1}$},
\smash{$N_{6.16}^{01}$},
\smash{$N_{6.18}^{0\beta1}$} 
from Turkowski's classification with the following canonical commutation relations:

\bigskip
\begin{tabular}{lll}
$N_{6.1\mathstrut}^{abcd\mathstrut}$ & $^{ac\ne 0\mathstrut}_{b^2+d^2\ne 0\mathstrut}$ &
$\arraycolsep=0ex\begin{array}{l}
[e_1,e_5]=ae_1,\ [e_2,e_5]=be_2,\  [e_4,e_5]=e_4,  \\[.5ex]
[e_1,e_6]=ce_1,\ [e_2,e_6]=d e_2,\ [e_3,e_6]=e_3;
\end{array}$
\\[3.5ex]
$N_{6.2\mathstrut}^{abc\mathstrut}$ &\scriptsize$a^2{+}b^2{\ne}0$ &
$\arraycolsep=0ex\begin{array}{l}
[e_1,e_5]=ae_1,\ [e_2,e_5]=e_2,\   [e_4,e_5]=e_3,\  \\[.5ex]
[e_1,e_6]=be_1,\ [e_2,e_6]=c e_2,\ [e_3,e_6]=e_3,\ [e_4,e_6]=e_4;
\end{array}$
\\[3.5ex]
$N_{6.13\mathstrut}^{abcd\mathstrut}$ & $^{a^2+c^2\ne 0\mathstrut}_{b^2+d^2\ne 0\mathstrut}$ &
$\arraycolsep=0ex\begin{array}{l}
[e_1,e_5]=ae_1,\ [e_2,e_5]=be_2,\ [e_3,e_5]=e_4,\ [e_4,e_5]=-e_3,   \\[.5ex]
[e_1,e_6]=ce_1,\ [e_2,e_6]=de_2,\ [e_3,e_6]=e_3,\ [e_4,e_6]=e_4;
\end{array}$
\\[3.5ex]
$N_{6.16\mathstrut}^{ab\mathstrut}$   &&
$\arraycolsep=0ex\begin{array}{l}
[e_2,e_5]=e_1,\ [e_3,e_5]=ae_3+e_4,\ [e_4,e_5]=-e_3+ae_4, \\[.5ex]
[e_1,e_6]=e_1,\ [e_2,e_6]=e_2,\      [e_3,e_6]=be_3,\ [e_4,e_6]=be_4;
\end{array}$
\\[3.5ex]
$N_{6.18\mathstrut}^{abc\mathstrut}$ &\scriptsize $b{\ne}0$  &
$\arraycolsep=0ex\begin{array}{l}
[e_1,e_5]=e_2,\ [e_2,e_5]=-e_1,\ [e_3,e_5]=ae_3+be_4,\ [e_4,e_5]=-be_3+ae_4, \\[.5ex]
[e_1,e_6]=e_1,\ [e_2,e_6]=e_2,\  [e_3,e_6]=ce_3,\      [e_4,e_6]=ce_4.
\end{array}$
\end{tabular}
\bigskip

\noindent
Since the above isomorphisms are not as obvious as for algebras of lower dimensions,
we present the necessary basis changes to $(e_1,e_2,e_3,e_4,e_5,e_6)$:
\begin{gather*}
\mbox{Case}~\ref{case17}\colon\quad
\big(P(0,e^{\beta t}),\,P(0,e^{-\beta t}),\,P(e^t,0),\,P(e^{-t},0),\,\tfrac12D^{\rm s}+\tfrac12D(1),\,\tfrac12D^{\rm s}-\tfrac12D(1)\big);\\
\mbox{Case}~\ref{case18}\colon\quad
\big(P(e^t,0),\,P(e^{-t},0),\,P(0,1),\,P(0,t),\,-D(1),\,D^{\rm s}\big)\quad\mbox{for}\quad \delta=0;\\
\mbox{Case}~\ref{case19}\colon\quad
\big(P(e^t,0),\,P(e^{-t},0),\,P(0,\cos\beta t),\,P(0,\sin\beta t),\,\beta^{-1}D(1),\,D^{\rm s}\big);\\
\mbox{Case}~\ref{case20}\colon\quad
\big(P(0,1),\,P(0,t),\,P(\sin t,0),\,P(\cos t,0),\,-D(1),\,D^{\rm s}\big)\quad\mbox{for}\quad \delta=0;\\
\mbox{Case}~\ref{case21}\colon\quad
\big(P(\cos t,0),\,P(\sin t,0),\,P(0,\cos\beta t),\,P(0,\sin\beta t),\,D(1),\,D^{\rm s}\big).
\end{gather*}
The Lie algebras from Turkowski's classification appearing in the consideration
are not isomorphic to each other,
including the pairs of algebras from the same series with different values of the parameter~$\beta$
within ranges indicated in the corresponding cases of Theorem~\ref{thm:GroupClassificationOf2DSWEs1}.
The claim on such pairs was checked by direct computation in {\sf Maple}.
This is why Cases~\ref{case17}, \ref{case18}$_{\delta=0}$, \ref{case19}, \ref{case20}$_{\delta=0}$ and~\ref{case21}
are inequivalent with respect to point transformations.
Moreover, the parameter~$\beta$ in Cases~\ref{case17}, \ref{case19} and~\ref{case21} cannot be gauged further.

Analyzing the classification cases listed in Theorem~\ref{thm:GroupClassificationOf2DSWEs1},
the additional equivalences among them that are found in this section
and the above consideration of the necessary algebraic conditions for their inequivalence,
we can suppose that the following assertion holds.

\begin{conjecture}\label{con:2DSWEsGroupClassificationWrtEquivGroupoid}
A complete list of inequivalent (up to all admissible transformations) Lie-symmetry extensions in the class~\eqref{eq:2DSWEs} is exhausted by
Cases~\ref{case1}, \ref{case2}, \ref{case3}, \ref{case4}, \ref{case5}, \ref{case6a}, \ref{case7}$_{\delta=0}$, \ref{case8}, \ref{case9},
\ref{case10}, \ref{case11}, \ref{case12a}, \ref{case13}$_{\delta=0}$, \ref{case14}$_{\delta=0}$,
\ref{case15}$_{\delta=0}$, \ref{case16a}$_{\delta=0}$, \ref{case17}, \ref{case18}$_{\delta=0}$, \ref{case19}, \ref{case20}$_{\delta=0}$, \ref{case21}
and~\ref{case22a} of Theorem~\ref{thm:GroupClassificationOf2DSWEs1}.
\end{conjecture}

To prove this conjecture, we need to complete the verification of inequivalence of cases
within the sets of cases 
$\{\ref{case1}_{\nu=2},\ref{case3}\}$, 
$\{\ref{case1}_{\nu\ne2},\ref{case2},\ref{case4},\ref{case5}\}$ 
and $\{\ref{case10},\ref{case11}\}$
as well as of impossibility of further gauging of constant parameters remaining in some cases.
This could be done via constructing the equivalence groupoid of the class~\eqref{eq:2DSWEs},
which is a nontrivial and cumbersome problem.
It is quite difficult to prove even principal properties of admissible transformations of the class~\eqref{eq:2DSWEs},
which can be conjectured after analyzing the equivalence transformations of this class,
the Lie symmetries of equations from this class and
the three obtained $G^\sim$-inequivalent families of admissible transformations. 
These properties include 
the affineness with respect to the dependent variables,
the fiber-preservation, i.e., the projectability to the space with the coordinates $(t,x,y)$,
as well as the projectability to the space with the coordinate $t$.
We can also conjecture the explicit structure of the equivalence groupoid of the class~\eqref{eq:2DSWEs}.

\begin{conjecture}\label{con:2DSWEsEquivGroupoid}
$G^\sim$-inequivalent non-identity admissible transformations of the class~\eqref{eq:2DSWEs}
that are independent up to inversion and composing with each other 
and with admissible transformations generated by Lie symmetries of systems from this class
are exhausted by the three families found in this section.
\end{conjecture}

\section{Conclusion}\label{sec:ConclusionsSWE}

We solved the group classification problem for the class~\eqref{eq:2DSWEs}
of two-dimensional shallow water equations with variable bottom topography.
The result is summarized in Theorem~\ref{thm:GroupClassificationOf2DSWEs1}.

Applying the algebraic method, we first construct the generalized equivalence group $G^\sim$ of the class~\eqref{eq:2DSWEs},
which is presented in Theorem~\ref{2d_gen_equiv}
and which is a necessary ingredient for solving the group classification problem.
Note that the generalized equivalence group of the class~\eqref{eq:2DSWEs} coincides with its usual equivalence group.

The integration of the system of determining equations for the components of Lie symmetry vector fields,
which is quite complicated in this case, required the application of the advanced method of furcate splitting.
This method was additionally optimized via reducing
the study of compatibility of template-form equations for the arbitrary element~$b$
to checking whether the set of vector fields associated with these equations
is closed with respect to the Lie bracket of vector fields.
In the course of the classification, we continuously used transformations from the equivalence group~$G^\sim$
for gauging various constants involved in the specific values of the arbitrary element~$b$,
which leads to a significant simplification of computations.

One more complication of the group classification for the class~\eqref{eq:2DSWEs} is that
this class is not normalized and even not semi-normalized.
In other words, this class possesses admissible point transformations
which cannot be decomposed 
into those generated by equivalence transformations of the class
and those generated by point symmetry transformations of systems belonging to it.
Such admissible transformations give rise to additional point equivalences
among the $G^\sim$-inequivalent classification cases listed in Theorem~\ref{thm:GroupClassificationOf2DSWEs1}.
In Section~\ref{sec:AdditionalEquivalencesSWE} we found
three families of $G^\sim$-inequivalent independent admissible transformations of the above kind in the class~\eqref{eq:2DSWEs}
and presented
the corresponding additional equivalences within the group-classification list for this class up to the $G^\sim$-equivalence.
Moreover, for all the pairs of listed cases that could be inequivalent to each other with respect to point transformations,
we checked their inequivalences via comparing the structure of the corresponding maximal Lie invariance algebras,
except for the inequivalences within three small subsets of cases with two- or three-dimensional maximal Lie invariance algebras.
This allowed us to conjecture the group classification of the class~\eqref{eq:2DSWEs} up to its equivalence groupoid~$\mathcal G^\sim$.

In future work, we plan to continue our study of shallow water equations to extend and generalize the obtained results.
The presented classification of Lie symmetries
of systems of two-dimensional shallow water equations with variable bottom topography
provides the basis for a wide research program for these systems 
within the framework of group analysis of differential equations, 
including classifications of admissible transformations, invariant solutions, 
generalized symmetries, cosymmetries, local conservations laws and Hamiltonian structures.

The first natural step in the further study is to describe 
the equivalence groupoid~$\mathcal G^\sim$ of the class~\eqref{eq:2DSWEs},
proving Conjecture~\ref{con:2DSWEsEquivGroupoid}.
If this conjecture is proved,
then the proof of Conjecture~\ref{con:2DSWEsGroupClassificationWrtEquivGroupoid} will be straightforward.
Moreover, the class~\eqref{eq:2DSWEs} will then give,
in addition to the class studied in~\cite[Section~X]{opan17a} and in~\cite{opan19a},
one more example of a class that is not semi-normalized
but in which all admissible transformations not generated by equivalence transformations
are still related to Lie-symmetry extensions.
An interesting question about the structure of~$\mathcal G^\sim$ is
how many $G^\sim$-inequivalent maximal conditional equivalence groups of the class~\eqref{eq:2DSWEs} exist;
see related definitions in~\cite{popo06Ay,popo10Ay}.

An obvious possibility for using Lie symmetries of systems from the class~\eqref{eq:2DSWEs} 
is given by Lie-symmetry reductions of these systems and the successive construction of their exact invariant solutions.
Such solutions can be employed to testing numerical schemes for the shallow water equations.
Moreover, Lie symmetries themselves may be applied 
for designing invariant parameterization and numerical schemes~\cite{bihlo2017a,popo12a}.
The detected additional point equivalences within the classification list are of great relevance here
since they allow one to avoid repeating the construction of exact solutions for systems
that are similar to simpler systems with respect to point transformations. 
We have already started to carry out
the Lie-reduction procedure and the construction of invariant parameterization schemes 
for the system of shallow water equations with the flat bottom topography.

The study of zeroth-order conservations laws of systems from the class~\eqref{eq:2DSWEs} 
was initiated in~\cite{akse18b,atam19a} but the results obtained there were preliminary or not exhaustive.
Extending the optimized version of the method of furcate splitting 
suggested in Section~\ref{sec:ProofClassification} to conservation laws 
and following the solution of the group classification problem in the present paper, 
in~\cite{bihlo2019b} we have classified zeroth-order conservation laws 
of systems from the class~\eqref{eq:2DSWEs} up to $G^\sim$-equivalence 
and indicated all additional point equivalences between the listed cases 
of extensions of the space of zeroth-order conservation laws. 
We have also constructed minimal generating sets of zeroth-order conservation laws 
for these systems under the action of corresponding Lie-symmetry groups therein. 
Moreover, generalizing the well-known Hamiltonian representation for 
the shallow water equations with flat bottom topography, 
we have derived the Hamiltonian representations for all systems from the class~\eqref{eq:2DSWEs}. 
This has allowed us to connect the classification results of~\cite{bihlo2019b} and the present paper 
via the Noether relation between Hamiltonian symmetries and conservation laws. 
Due to knowing the Hamiltonian representations, 
we have shown that each system from the class~\eqref{eq:2DSWEs} admits,  
along with zeroth-order conservation laws, 
an infinite-dimensional space of first-order conservation laws 
associated with distinguished (Casimir) functionals of the Hamiltonian operator, 
which is common for all such systems. 

We conjecture that for any system from the class~\eqref{eq:2DSWEs},
its generalized symmetries are in fact exhausted, 
up to the equivalence of generalized symmetries, by its Lie symmetries.  
If this conjecture is proved, then the Noether relation will imply 
that for any system from the class~\eqref{eq:2DSWEs}
the entire space of its local conservation laws is spanned 
by the above zeroth- and first-order conservation laws.

\section*{Acknowledgments}
The authors are grateful to the reviewers for useful suggestions. 
The authors also thank Michael Kunzinger, Dmytro Popovych and Galyna Popovych for helpful discussions.
The research of AB and NP was undertaken, in part, thanks to funding from the Canada Research Chairs program,
the InnovateNL LeverageR{\&}D program and the NSERC Discovery program.
The research of ROP was supported by the Austrian Science Fund (FWF), projects P25064 and P30233.

\footnotesize

\end{document}